\PassOptionsToPackage{table}{xcolor}
\documentclass[sigconf]{acmart}

\setcopyright{none}
\acmConference[WWW '27]{The ACM Web Conference 2027}{May 10--14, 2027}{Dublin, Ireland}
\acmYear{2027}
\acmISBN{}
\acmDOI{}
\usepackage{tabularx}
\usepackage{array}
\usepackage{colortbl}
\usepackage{algorithmic}
\usepackage{algorithm}
\usepackage{booktabs}
\usepackage{multirow}
\usepackage{cleveref}
\usepackage{subcaption}
\usepackage{enumitem}
\usepackage{xspace}
\usepackage{pifont}
\usepackage{etoolbox}
\usepackage[most]{tcolorbox}
\usepackage{makecell}

\newtheorem{invariant}{Invariant}
\newtheorem{assumption}{Assumption}
\crefname{assumption}{Assumption}{Assumptions}
\Crefname{assumption}{Assumption}{Assumptions}
\newtheorem{remark}{Remark}

\definecolor{paperaccent}{RGB}{44,87,118}
\definecolor{paperaccentsoft}{RGB}{232,239,245}
\definecolor{takeawayblue}{RGB}{236,242,247}
\definecolor{takeawayborder}{RGB}{116,143,165}
\definecolor{theorycol}{RGB}{228,235,241}
\definecolor{leadtagblue}{RGB}{225,234,244}
\definecolor{leadtagteal}{RGB}{224,239,232}
\definecolor{leadtaggold}{RGB}{247,235,201}
\definecolor{leadtagrose}{RGB}{244,225,228}
\colorlet{tagsetup}{leadtagblue}
\colorlet{tagsemantic}{leadtaggold}
\colorlet{tagdefense}{leadtagteal}
\colorlet{tagthreat}{leadtagrose}
\definecolor{AstralBlue}{RGB}{232,239,245}
\definecolor{DatasetGray}{RGB}{245,245,245}
\definecolor{MutedText}{gray}{0.45}

\newcommand{\astral}{\textsc{Astral}\xspace}

\newcommand{\evaltablefont}{\small}

\newcommand{\papertablesetup}{%
  \evaltablefont
  \setlength{\tabcolsep}{5pt}%
  \renewcommand{\arraystretch}{1.08}%
}

\newcommand{\astralcell}[1]{\cellcolor{paperaccentsoft}\textbf{#1}}
\newcommand{\astralsoft}[1]{\cellcolor{paperaccentsoft!72}#1}
\newcommand{\commitqc}{\mathsf{CommitQC}}
\newcommand{\semqc}{\mathsf{SemQC}}
\newcommand{\ncc}{\mathsf{NCC}}
\newcommand{\stableqc}{\mathsf{StableQC}}
\newcommand{\counterqc}{\mathsf{CounterQC}}
\newcommand{\crqc}{\mathsf{CRQC}}
\newcommand{\daqc}{\mathsf{DAQC}}
\newcommand{\vstruct}{V_{\mathit{struct}}}
\newcommand{\vinfer}{V_{\mathit{infer}}}

\newcommand{\afs}{\mathsf{AFS}}
\newcommand{\ds}{\mathsf{DS}}
\newcommand{\wid}{\mathsf{wid}}

\newcommand{\eg}{\textit{e.g.}\xspace}

\newtcbox{\leadpill}[1][]{%
  on line,
  boxrule=0.3pt,
  arc=2pt,
  left=3pt,
  right=3pt,
  top=1pt,
  bottom=1pt,
  colback=leadtagblue,
  colframe=paperaccent!18,
  coltext=black,
  boxsep=0pt,
  #1
}
\newcommand{\leadtag}[3]{\par\smallskip\noindent\smash{\leadpill[colback=#2,colframe=#2!55!black]{\textbf{#1}}}\ #3}

\newcommand{\takeawaybox}[1]{%
  \begin{tcolorbox}[
    enhanced,
    colback=takeawayblue, colframe=takeawayborder,
    boxrule=0.6pt, arc=2pt,
    left=4pt, right=4pt, top=3pt, bottom=3pt,
    before skip=4pt, after skip=6pt
  ]
  \textbf{Takeaway.} #1
  \end{tcolorbox}
}
\newcommand{\cmark}{\textcolor{green!60!black}{\ding{51}}}
\newcommand{\xmark}{\textcolor{red}{\ding{55}}}
\newcommand{\pmark}{{\color{orange}$\boldsymbol{\sim}$}}

\newenvironment{packeditemize}{%
  \begin{list}{$\bullet$}{%
    \setlength{\labelwidth}{4pt}%
    \setlength{\itemsep}{0pt}%
    \setlength{\leftmargin}{\labelwidth}%
    \addtolength{\leftmargin}{\labelsep}%
    \setlength{\parindent}{0pt}%
    \setlength{\listparindent}{\parindent}%
    \setlength{\parsep}{0pt}%
    \setlength{\topsep}{1pt}}}{\end{list}}

\setlist[itemize]{leftmargin=*,topsep=2pt,itemsep=1pt}
\AtBeginEnvironment{table}{\papertablesetup}
\AtBeginEnvironment{table*}{\papertablesetup}

\let\origus\_\renewcommand{\_}{\origus\allowbreak}
\begin{document}

\title{Who Verifies the Verifiers? Accountable Semantic Consensus for AI-Mediated Web Information Verification}

\author{Yanna Jiang}
\authornote{Yanna Jiang and Mingyuan Li contributed equally to this work.}
\affiliation{\institution{University of Technology Sydney}\country{Australia}}

\author{Mingyuan Li}
\authornotemark[1]
\affiliation{\institution{ELLIS Institute Finland}\institution{University of Turku}\country{Finland}}

\author{Xu Wang}
\affiliation{\institution{University of Technology Sydney}\country{Australia}}

\author{Guangsheng Yu}
\authornote{Corresponding author.}
\affiliation{\institution{University of Technology Sydney}\country{Australia}}

\begin{abstract}
Large language models (LLMs) increasingly decide whether Web evidence supports online claims, often by voting over several assessors whose errors are presumed independent and benign.
On the open Web, retrieved pages can steer the assessors that read them, and assessors from one model family fail together.
We prove that even for independent assessors whose every output sequence is possible on both valid and invalid claims, any protocol that decides only from their outputs and always reaches an agreed verdict never rejects invalid claims or rejects valid ones with positive probability.

We present \astral, which pins all assessors to one committed snapshot of the Web evidence. With up to $f$ assessors faulty, \astral accepts a verdict only when $2f{+}1$ assessors attest to it above a confidence threshold and no signed challenge stands.
On five public benchmarks with two of seven assessors adversarial, \astral lowers pooled attack success to 21.5--53.1\%, against 54.7--77.8\% for majority voting.
The cost is coverage, as without attack 16.0--61.7\% of claims receive a binary verdict and the rest are disputed for audit.
Our code runs at \url{https://anonymous.4open.science/r/Astral-55F1}.
\end{abstract}



\maketitle

\section{Introduction}\label{sec:introduction}

Large language models (LLMs) mediate much of what Web users read about online claims, and LLM verdicts reach users with the authority of the services that display them.
Search engines answer queries with generated summaries of retrieved pages, and LLMs are proposed and deployed as fact-checking assistants that flag or annotate misleading claims~\cite{factualitychallenges,misinfollm}.
In each pipeline, an LLM decides whether the evidence it retrieved supports a Web claim.
LLMs also judge other systems' outputs, and agreement with human assessors is the main evidence that LLMs are valid judges~\cite{llmjudge}.
Because these AI systems help establish trust on the Web, we ask how to make them trustworthy.
The work fits the Responsible Web track by addressing Web information integrity and accountability for AI systems that assess Web content.

\leadtag{Problem}{tagthreat}{Deployments use several LLM assessors to avoid a single point of failure and aggregate their verdicts by majority vote or confidence-weighted refinement~\cite{multiagentdebate,cpwbft}.}
Aggregation helps only when assessor errors are independent and benign, and the open Web violates both conditions.
Retrieval exposes assessors to content written to mislead them, including indirect prompt injection in Web pages~\cite{promptinjection}.
Assessors from one model family share failure modes and misjudge the same claims.
Some assessors may also be operated by parties with a stake in the verdict.
Label aggregation then turns the errors of a minority into an authoritative Web decision.
\looseness=-1 With two of seven assessors adversarial, confidence-weighted aggregation accepts 20--36\% of refuted synthetic multi-source claims because one assessor claiming high confidence shifts the others to its label (\S\ref{sec:eval:rq2}).

\leadtag{Why stronger voting fails}{tagsemantic}{Stronger voting rules do not remove this weakness.}
Consider any protocol that decides claims only from assessors' outputs and always reaches an agreed verdict.
When independent assessors can produce every output sequence on both valid and invalid claims, such a protocol provably rejects valid claims with positive probability or rejects no invalid claim (Proposition~\ref{prop:impossibility}).
For majority voting, this probability grows with each assessor's error rate.
The proposition also binds the quorum of \astral.
Trust in AI-mediated Web verification therefore requires bounded residual risk and the ability to contest a verdict.
We study this requirement through three research questions.
\begin{itemize}
  \item \textbf{RQ1.} How reliably do aggregation rules decide Web claims when heterogeneous assessors disagree benignly?
  \item \textbf{RQ2.} How fast does verification integrity degrade when a minority of assessors fail systematically or act adversarially?
  \item \textbf{RQ3.} Is semantic consensus a usable integrity primitive for AI-mediated Web verification, and what does it cost?
\end{itemize}
Semantic consensus certifies a verdict only when assessors independently attest it using the same snapshot of evidence and the verdict remains open to challenge.

\leadtag{Approach}{tagdefense}{\looseness=-1 We propose \astral (\textbf{A}ccountable \textbf{S}emantic \textbf{T}riage, \textbf{R}easoning \textbf{A}ttestation \& \textbf{L}iveness), a semantic consensus layer that treats each AI verdict on a Web claim as a protocol object with a lifecycle (\Cref{fig:overview}).}
An ordering layer (L1) fixes the verdict proposal.
Because Web pages change, assessors would see different evidence if they retrieved pages at different times.
An evidence layer (L2) therefore requires every assessor to evaluate the same committed snapshot of retrieved Web evidence.
With up to $f$ assessors Byzantine, a semantic layer (L3) collects $2f{+}1$ attestations that pass a confidence threshold and opens a window in which any assessor can file a signed challenge.
A platform may act on the ordered verdict immediately, and \astral upgrades it to accountable acceptance only after attestation succeeds and no challenge stands.
Contested verdicts are marked as disputed and are left unresolved by vote.

\leadtag{Contributions}{tagsetup}{We make four contributions.}

\begin{packeditemize}
  \item \textit{Assessor failure in Web verification.} We identify failures by assessors as a weakness of AI-mediated Web verification and prove that when the output sequences of independent assessors never separate valid from invalid claims, a protocol that rejects invalid claims also rejects valid ones with positive probability (\S\ref{sec:threat}).
  \item \looseness=-1 \textit{Accountable semantic consensus.} \astral commits a snapshot of Web evidence, counts attestations that pass a confidence threshold, and keeps verdicts challengeable. We bound residual false acceptance by assessor error, error correlation, and challenge success (\S\ref{sec:design}).
  \item \textit{Robustness on Web claims.} We evaluate seven open-weight assessors on five public benchmarks. With two adversarial assessors, \astral lowers pooled attack success over four attacks to 21.5--53.1\%, against 54.7--77.8\% for majority voting and 70.8--85.2\% for confidence-weighted aggregation. Without attack it returns binary verdicts on 16.0--61.7\% of claims, certifies 6.8--30.3\% with a StableQC, and sends the rest to \texttt{DISPUTED} for audit (\S\ref{sec:evaluation}).
  \item \textit{Assessor diversity.} We observed joint false acceptances in a committee from one model family but none in committees spanning three or more families. On the public benchmarks, single-model committees range from 1.4--8.0\% to 57.5--90.1\% pooled attack success depending on the backend. We therefore identify diversity as a hedge against a weak backend, which does not guarantee lower risk than the strongest single model.
\end{packeditemize}

\section{AI-Mediated Web Verification}\label{sec:threat}

\begin{figure*}[t]
  \centering
  \includegraphics[width=0.86\linewidth]{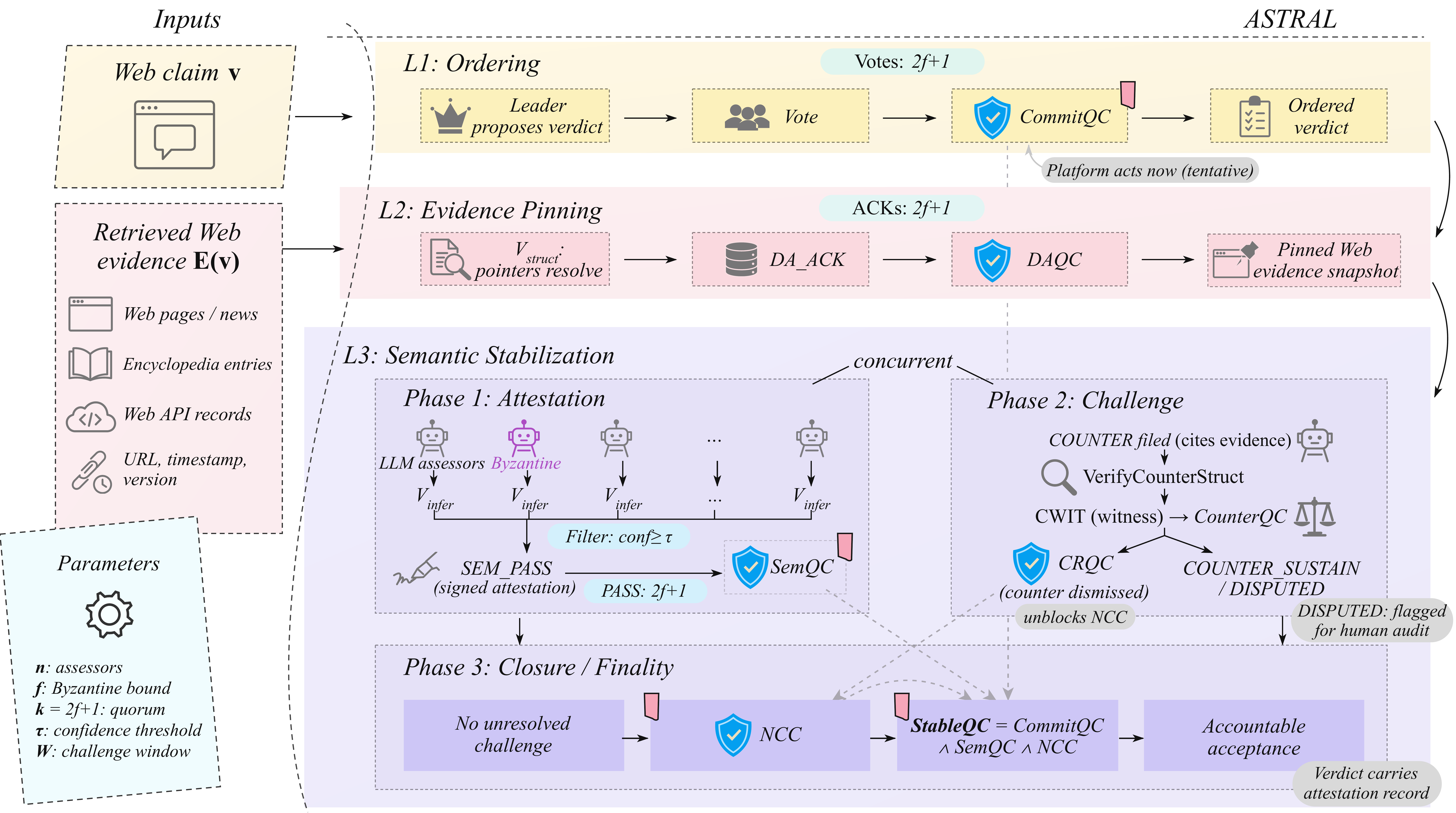}
  \caption{%
    \astral as a verification layer for Web claims.
    \textbf{L1} orders the verdict proposal and yields $\commitqc$, on which a platform can act immediately.
    \textbf{L2} pins the retrieved Web evidence, because assessors sign \texttt{DA\_ACK} only after the structural check $\vstruct$ passes and $\daqc$ certifies one committed evidence snapshot.
    \textbf{L3} runs attestation concurrently with a challenge window. Attestation aggregates \texttt{SEM\_PASS} messages into $\semqc$, and during the window assessors file and triage \texttt{COUNTER}s.
    With no unresolved challenge, $\ncc$ forms and $\stableqc = \commitqc \wedge \semqc \wedge \ncc$ certifies accountable acceptance.
  }
  \label{fig:overview}
\end{figure*}

\subsection{Setting}\label{sec:threat:system}

A \emph{Web verification task} $D$ pairs a Web claim with pointers to evidence retrieved for it, such as news articles, encyclopedia entries, or vulnerability records.
A committee of $n$ LLM assessors $p_1, \ldots, p_n$, each identified by a public key, decides $D$.
Each $p_i$ queries its LLM backend, stochastic oracle $O_i$, to run $\vinfer(D)$, which stochastically estimates the inference policy.
The check returns a label in $\{\mathsf{SUPPORT}, \mathsf{REJECT}\}$, a confidence in $[0,1]$, and a pointer to the evidence used.
Retrieved evidence is written to an append-only, content-addressed store $\afs$, whose Merkle root is committed with the proposal for $D$.
Every assessor therefore evaluates the same snapshot even when live pages change or disappear after retrieval.
A claim is semantically valid when the committed policy assigns the claim confidence at or above the acceptance threshold using this snapshot.
Membership is static, and up to $f$ assessors may be Byzantine (arbitrarily faulty), with $n \geq 3f+1$.
The network is partially synchronous. After an unknown global stabilization time (GST), messages arrive within the known bound $\Delta$ (\Cref{asm:dls}).
An ordering protocol can provide the ordering layer if it exposes an external-validity hook~\cite{pbft,hotstuff,tendermint} and meets the assumptions of \textbf{Appendix~\ref{app:formalization}}.
\looseness=-1 $\commitqc$ lets the platform act on an ordered verdict immediately, and $\stableqc$ later certifies accountable acceptance (\S\ref{sec:design}).

\subsection{How Assessors Fail on the Web}\label{sec:threat:faults}

Each of four failure sources in Web verification maps to a quantity controlled by the protocol or its analysis.
An honest assessor makes a \emph{model error} by hallucinating or misreading.
Model error contributes to each assessor's false-accept rate $\epsilon$, which we measure as $\hat{\epsilon}(\tau)$ at confidence threshold $\tau$ (\S\ref{sec:design:probability}).
\emph{Content-induced error} arises when a retrieved page is written to mislead, for instance through an injected instruction or fabricated corroboration~\cite{promptinjection}.
Such a page raises the false-accept rate of every assessor that reads it to some $\epsilon' \geq \epsilon$.
\emph{Correlated error} arises when assessors share a model family or read the same misleading page and therefore misjudge the same claims.
We account for correlated error with the correlation factor $\rho_{k_h}$ (\Cref{eq:rho}), which is at most 1 when assessor errors are independent and grows as assessors share failures.
\emph{Adversarial assessors} are operated by a party with a stake in the verdict or steered into arbitrary behavior by content.
Up to $f$ such assessors may collude, suppress evidence, equivocate selectively, or otherwise behave arbitrarily.
We design the quorum, certificate, and challenge rules of \astral for this fault model.

Content that steers more than $f$ assessors to the same wrong label lies outside any quorum guarantee.
In the analysis this case appears as $\epsilon'$ close to 1 and $\rho_{k_h} \gg 1$.
The residual-risk estimate then degrades without overstating protection.
We confine all four failure sources to LLM outputs through the following rule.

\begin{invariant}[Deterministic Wrapper]\label{inv:wrapper}
Protocol state transitions never invoke an LLM\@.
All consensus messages (\texttt{VOTE}, \texttt{NO\_CHALLENGE}, \texttt{CWIT}, and the like) are generated by deterministic code from verified local state.
LLM inference is invoked only to generate proposal content, to produce \texttt{SEM\_PASS} attestations via $\vinfer$, and to produce \texttt{COUNTER\_SUSTAIN} or semantic \texttt{COUNTER\_REJECT} votes during counter triage.
\end{invariant}

\noindent An honest assessor misled by content can therefore produce wrong labels only in these three message types, and it cannot cause equivocation, invalid votes, or inconsistent protocol transitions.
A Byzantine assessor need not run the wrapper.

\subsection{Agreement Is Not Validity}\label{sec:threat:guarantees}

\looseness=-1 Adding assessors cannot by itself make a Web verdict valid.

\begin{proposition}[False Rejection under Shared Support]\label{prop:impossibility}
Let a protocol $\Pi$ read the input $D$ only through outputs of independent stochastic oracles $O_i$ in a common finite set, with queries, stopping, and verdict that depend only on earlier outputs and on coins independent of $D$ and of the oracles.
Assume \emph{shared support}, meaning that for every coin value and after every reachable history each next output has positive probability on every semantically non-trivial input.
Let $\Pi$ achieve agreement, terminate with probability 1, and reject with positive probability some input in its \emph{rejection coverage} $F_{\mathit{inv}}$, the set of non-trivial invalid inputs it can reject.
Then $\Pi$ falsely rejects every non-trivial valid input $D'$ with probability
\begin{equation}\label{eq:impossibility}
P_{\mathit{err}}(\Pi, D') > 0.
\end{equation}
\emergencystretch=1em If each of $n$ assessors is wrong on $D'$ independently with probability $e$, majority reconciliation falsely rejects $D'$ with probability $\sum_{i > n/2} \binom{n}{i} e^i (1 - e)^{n-i}$, which is $2.7 \times 10^{-3}$ for $n = 7$ and $e = 0.1$.
\end{proposition}

\noindent\emph{Proof sketch.}
With positive probability over the coins, some finite output history leads $\Pi$ to reject an input in $F_{\mathit{inv}}$.
Because $\Pi$ sees $D$ only through outputs, the same history leads to $\mathsf{INVALID}$ on a valid $D'$, where shared support gives it positive probability.
\textbf{Appendix~\ref{app:proof:tradeoff}} gives the full proof.

The proposition fixes what a Web verification layer can certify.
\astral certifies \emph{accountable acceptance}, meaning a verdict that was independently attested on a pinned evidence snapshot, survived a challenge window, and carries a signed certificate chain recording who attested and who challenged.
Its guarantees hold at three levels.
Ordering agreement is inherited unconditionally from L1.
Evidence consistency is deterministic, since $\daqc$ binds every attestation to one committed snapshot that at least $f{+}1$ honest assessors hold.
\looseness=-1 Semantic acceptance is distributional, with a residual false-acceptance rate bounded by $\delta$ over the task distribution (\Cref{eq:delta}).

\section{The \astral Protocol}\label{sec:design}

\astral handles a Web verdict in three layers with different guarantees (Fig.~\ref{fig:overview}).
A platform acts at $\commitqc$, while \astral runs L3 entirely after commitment to certify accountable acceptance ($\stableqc$) or flag the verdict as disputed.

\subsection{Ordering and Evidence Pinning}\label{sec:design:l1}\label{sec:design:daqc}

\looseness=-1 L1 provides $\mathsf{Commit}(h) = \mathit{val}_h$, a totally ordered sequence of verdict proposals with agreement (unconditional for deterministic BFT) and post-GST liveness.
An ordering protocol that exposes an external-validity predicate and meets the assumptions of \textbf{Appendix~\ref{app:formalization}} suffices, including deterministic BFT~\cite{hotstuff,pbft,aptosdiembft} and probabilistic chains with $k$-deep confirmation (\textbf{Appendix~\ref{app:hotstuff}}).
The fault threshold $n \geq 3f+1$ is inherited from L1.
Ordering treats the proposal as opaque bytes and says nothing about whether the claim is supported.

Each proposal carries a \emph{decision object}
\begin{align}\label{eq:decision}
D = (&\mathtt{claim},\, \mathtt{evid\_pointers},\, \mathtt{lineage\_pointers},\, \notag\\
     &\mathtt{policy\_hash},\, \mathtt{afs\_root},\, H_{\mathit{ctx}}),
\end{align}
which binds a Web claim to pointers into the committed evidence store, the hash of an inference policy $\pi$, and a context hash $H_{\mathit{ctx}}$.
The policy assigns confidence to a claim given evidence.
A claim is semantically valid under $\pi$ when $\pi(E, \mathtt{claim}) \geq \tau$ for the committed evidence $E$.
Each assessor's $\vinfer$ estimates $\pi$ stochastically (\textbf{Appendix~\ref{app:form:decision}}).
Honest assessors can therefore disagree on claims whose confidence lies near $\tau$.

\smallskip
\noindent\textbf{Evidence pinning (L2).}
Web pages change, so retrieving at different times would give assessors different evidence and make their attestations incomparable.
After $\commitqc$, the leader broadcasts the pre-block $P$ and referenced artifacts.
$P$ carries the committed value with $\mathsf{afs\_root}$, the evidence pointers, and $H_{\mathit{ctx}}$.
Each assessor runs the deterministic admission check $\vstruct$ to verify that every evidence pointer resolves against $\mathsf{afs\_root}$ and that the policy hash is valid.
The check verifies that $H_{\mathit{ctx}}$ is consistent with the pointers and that recorded tool outputs match their hashes.
An assessor whose check passes signs $\texttt{DA\_ACK}(\ds, H(P), \mathsf{afs\_root})$, and $2f{+}1$ signatures form $\daqc$.
As with the certificates of availability in Narwhal~\cite{narwhal}, $2f{+}1$ signatures guarantee that at least $f{+}1$ honest assessors hold the evidence.
L3 begins only after $\daqc$, so every honest assessor evaluates the snapshot fixed by the committed Merkle root.
All honest assessors agree on the outcome of $\vstruct$ because it is deterministic.
The check is purely restrictive, so it cannot introduce an L1 safety violation.
A failed check makes the committed block a no-op for L2 and L3.
Unavailable evidence stalls $\daqc$, which costs liveness but never safety.

\subsection{Semantic Stabilization}\label{sec:design:l3}

\emergencystretch=3em L3 moves a committed verdict from tentative agreement ($\commitqc$) to accountable acceptance ($\stableqc$).
Attestation (Phase~\textcircled{1}) and the challenge window (Phase~\textcircled{2}) run concurrently after $\daqc$, and closure (Phase~\textcircled{3}) follows the window.

\smallskip
\noindent\textbf{Phase~\textcircled{1}: attestation.}\label{sec:design:attestation}
\looseness=-1 Each assessor runs $\vinfer(D)$ on the pinned snapshot.
An assessor that obtains $\mathsf{SUPPORT}$ with confidence at least $\tau$ signs a \texttt{SEM\_PASS} attestation, and $k = 2f{+}1$ attestations from distinct assessors form $\semqc$.
We count attestations whose confidence meets $\tau$ equally because weighting by confidence can let a single Byzantine assessor reporting high confidence set the outcome (\S\ref{sec:eval:rq2}).
The $2f{+}1$ distinct attesters include at least $f{+}1$ honest assessors (Lemma~\ref{lem:honestsubset}). Thus the $f$ Byzantine assessors cannot form $\semqc$ alone, whatever confidence they report.
Setting $k = 2f{+}1$ maximizes how many honest assessors attest to each accepted verdict.
Reducing $k$, down to $f{+}1$, tolerates more false rejections by honest assessors at the cost of a weaker residual-risk bound (\textbf{Appendix~\ref{app:protocol:threshold}}).

\smallskip
\noindent\textbf{Phase~\textcircled{2}: challenge window.}\label{sec:design:challenge}
\looseness=-1 After $\commitqc$, a challenge window of duration $W(v) \geq 2\Delta$ opens, during which any assessor can file a signed \texttt{COUNTER} that references evidence in $\afs$ contradicting the claim.
Because assessors admit only counters filed before $W(v) - \Delta$, every admitted counter reaches all honest assessors before the window closes post-GST.
The window adapts the fraud proofs of optimistic rollups~\cite{rollupsurvey,arbitrum} to non-deterministic verification.
A \emph{deterministic fraud proof} re-executes a recorded tool call and can be checked by every assessor, whereas a \emph{semantic counter} argues from evidence and must be triaged with $\vinfer$.
Accountability is cryptographic, and no staking is involved.
The window grows linearly in the view number as $W_0 \cdot v$, where $W_0$ is added at each view (\Cref{alg:consensus}).
This linear growth guarantees coverage while the network remains unstable before GST.
A lower cap on $W(v)$ gives faster acceptance but leaves less time for honest challenges (\textbf{Appendix~\ref{app:protocol}}).

A counter first passes \texttt{VerifyCounterStruct}, which deterministically rejects malformed and duplicate counters.
Then $2f{+}1$ witness signatures (\texttt{CWIT}) form a $\counterqc$.
A valid fraud proof marks the claim \textsc{refuted} and blocks $\stableqc$ permanently.
Each assessor triages a semantic counter with $\vinfer$ and votes \texttt{COUNTER\_SUSTAIN} or \texttt{COUNTER\_REJECT}.
Then $2f{+}1$ \texttt{COUNTER\_REJECT} votes form $\crqc$, which dismisses the counter.
If the counter is not dismissed by the triage timeout $T_{\mathit{arb}}$, the decision becomes \texttt{DISPUTED}.

\begin{definition}[Counter Lifecycle]\label{def:counterlifecycle}
A counter for $(h, H(b), \wid)$ is \emph{pending} once it is structurally valid and added to \textsf{unresolved\-Counters}.
It becomes \emph{resolved} when it reaches a terminal state ($\crqc$ rejection, \textsc{sustained}, or \textsc{disputed} timeout), and it is then \emph{pruned} from the Mandatory Inclusion enforcement set.
Only a $\crqc$ rejection removes it from $\mathsf{unresolvedCounters}$, the authoritative blocking criterion for $\ncc$ signing.
A \textsc{sustained} fraud proof is final for height $h$. A \textsc{disputed} flag is provisional, and a $\stableqc$ formed later replaces it.
\end{definition}

\emergencystretch=1em A structurally valid counter enters $\mathsf{unresolvedCounters}$ and leaves it only at a terminal state (Definition~\ref{def:counterlifecycle}), so its filer cannot withdraw it.
Mandatory Inclusion prevents leaders from dropping certified counters because assessors deny \texttt{DA\_ACK} to proposals that omit a known $\counterqc$ for an unresolved counter.
Unblocking $\ncc$ requires $2f{+}1$ \texttt{COUNTER\_REJECT} votes so that the $f$ Byzantine assessors alone can never dismiss a sound counter.
Each Byzantine assessor can file at most one structurally valid counter per window, and triaging each counter adds at most $T_{\mathit{arb}}$.
Hence, including data availability and inference time, stabilization after $\commitqc$ satisfies
\begin{equation}\label{eq:delay-bound}
T_{\text{stable}} \leq T_{\mathit{DA}} + T_{\mathit{infer}} + W(v) + f \cdot T_{\mathit{arb}}.
\end{equation}
\looseness=-1 Beyond this bound, unresolved counters trigger \texttt{DISPUTED}, and the platform suspends its $\commitqc$ action pending human audit.

\smallskip
\noindent\textbf{Phase~\textcircled{3}: closure.}\label{sec:design:closure}
When the window closes with no unresolved counter, assessors sign \texttt{NO\_CHALLENGE}.
Then $2f{+}1$ signatures form $\ncc$, and the verdict reaches semantic finality with $\stableqc = \commitqc \wedge \semqc \wedge \ncc$.
Without $\semqc$ at closure and with no counter pending or final, the verdict is rejected provisionally until a forwarded $\semqc$ arrives.
\looseness=-1 \Cref{tab:dualfinality} lists the four outcomes a platform sees.

\begin{table}[t]
\centering
\caption{Dual-finality outcomes of a Web verdict.}\label{tab:dualfinality}
\resizebox{\columnwidth}{!}{
\begin{tabular}{@{}llll@{}}
\toprule
\textbf{Level} & \textbf{Certificate} & \textbf{Guarantee} & \textbf{Platform action} \\
\midrule
Tentative & $\commitqc$ & L1 agreement & Act immediately \\
Stable & \astralcell{$\stableqc$} & \astralcell{Accountable acceptance} & \astralcell{Upgrade confidence} \\
Rejected & none & No quorum yet (provisional) & Withdraw action \\
Disputed & \astralsoft{none} & \astralsoft{Liveness preserved} & \astralsoft{Suspend, human audit} \\
\bottomrule
\end{tabular}
}
\end{table}

\subsection{Residual-Risk Bound}\label{sec:design:probability}

On invalid claims from a characterized distribution $\mathcal{D}$ of Web claims, $\epsilon$ bounds each honest assessor's false-accept rate, and we estimate it by $\hat{\epsilon}(\tau)$.
With $A_i$ the event that honest assessor $i$ falsely accepts and $H$ a fixed set of $n-f$ honest assessors, we define the correlation factor
\begin{equation}\label{eq:rho}
\rho_{k_h} := \max_{S \subseteq H,\, |S| = k_h} \Pr\!\left[\textstyle\bigcap_{i \in S} A_i\right] / \epsilon^{k_h},
\end{equation}
where $k_h := f+1$ is the number of attesters from $H$ that every $\semqc$ contains (Lemma~\ref{lem:honestsubset}).

\begin{lemma}[Honest Subset in $\semqc$]\label{lem:honestsubset}
$\semqc$ requires $k = 2f+1$ attestations from distinct assessors.
Only $f$ assessors lie outside $H$, so at least $k - f = k_h$ attestations come from assessors in $H$.
\end{lemma}

An invalid claim obtains $\semqc$ only in the event $Q$ that at least $k_h$ assessors in $H$ falsely accept.
Because $Q$ is the union of the joint false-acceptance events of the $\binom{n-f}{k_h}$ subsets of $H$ of size $k_h$, the union bound gives $\Pr[Q] \leq \binom{n-f}{k_h}\epsilon^{k_h}\rho_{k_h}$.
A counter that is filed in time, delivered, and upheld in triage blocks accountable acceptance. With $p_{\mathit{ch}}$ the probability of such a counter conditioned on $Q$,
\begin{equation}\label{eq:delta}
\delta(\epsilon, k_h) \leq \tbinom{n-f}{k_h} \cdot \epsilon^{k_h} \cdot \rho_{k_h} \cdot (1 - p_{\mathit{ch}}).
\end{equation}
An unconditional estimate of $p_{\mathit{ch}}$ is valid only when $Q$ and challenge failure are close to independent, and subtle invalidity can violate that independence (\textbf{Appendix~\ref{app:formalization}}).
Post-GST, $\counterqc$ and Mandatory Inclusion deliver every filed counter. If the $m$ honest checkers detect independently given $Q$ with probability $p_{\mathit{det}}$ each, at least one of them detects with probability $1 - (1 - p_{\mathit{det}})^m$, and triage must still uphold the counter.

\smallskip
\noindent\textbf{Diversity.}
Operators should configure a diversity predicate that requires attestations from $\geq d$ distinct oracle implementations to reduce $\rho_{k_h}$.
This predicate is an engineering recommendation (\Cref{app:protocol:diversity}).
Because $\rho_{k_h}$ multiplies the bound in \Cref{eq:delta}, correlated errors in a less diverse committee loosen $\delta$.
On the multi-source claims of \S\ref{sec:eval:setup}, a fixed group drawn from a single model family has an observed ratio $\hat{\rho}_s = 50$ (\Cref{tab:rho}).

\subsection{Security Properties}\label{sec:design:security}

\begin{theorem}[Agreement Safety]\label{thm:agreement}
No two honest assessors commit conflicting values at the same height.
\end{theorem}

\begin{theorem}[Stable Agreement]\label{thm:stableagreement}
\looseness=-1 If two honest assessors independently output $\mathsf{Certify}(h)$ for height $h$, they output the same value.
\end{theorem}

\begin{theorem}[\astral Liveness]\label{thm:liveness}
Assume \Cref{asm:a1,asm:a2,asm:a3,asm:a4,asm:a5,asm:a6}, post-GST operation, and $2f{+}1$ \texttt{SEM\_PASS} attestations that arrive in time.
Honest assessors falsely reject at rate $\epsilon_{\mathit{fn}}$, so for $n = 3f{+}1$ this count holds in expectation when $\epsilon_{\mathit{fn}} \leq f/(2f{+}1)$ and all $f$ Byzantine assessors attest in time.
Under worst-case withholding it holds with certainty only when $\epsilon_{\mathit{fn}} = 0$.
Then $\mathsf{Commit}(h)$ is emitted within $O(f)$ view-changes, as inherited from L1.
\looseness=-1 Unless a valid deterministic fraud proof refutes the claim, $\stableqc$ also forms within $O(2\Delta + W(v) + T_{\mathit{DA}} + T_{\mathit{infer}} + f \cdot T_{\mathit{arb}})$ after $\commitqc$, or the decision is flagged \texttt{DISPUTED}.
\end{theorem}

\begin{theorem}[Eventual Counter Dissemination]\label{thm:dissemination}
Post-GST, deterministic $\counterqc$ forwarding and Mandatory Inclusion guarantee that all honest assessors converge on the same counter set within $\Delta$.
Censoring counters delays progress by at most $O(f)$ view-changes.
\end{theorem}

\emergencystretch=1em Agreement safety is inherited from L1 because $\vstruct$ filters blocks after commitment.
Stable agreement follows because $\commitqc$, $\semqc$, and $\ncc$ must reference the same committed block hash $H(b)$.
Liveness follows because L1 progresses, attestation and triage take bounded time, and a timeout flags the decision \texttt{DISPUTED}.
\textbf{Appendix~\ref{app:censorship}} proves Proposition~\ref{prop:impossibility} and sketches the proof of \Cref{thm:liveness}.
The anonymous artifact repository (\url{https://anonymous.4open.science/r/Astral-55F1}) contains the full formal proofs of \Cref{thm:agreement,thm:stableagreement,thm:liveness,thm:dissemination}.
\textbf{Appendix~\ref{app:protocol:pseudocode}} gives pseudocode and invariants that keep L3 state out of L1 voting (\Cref{inv:isolation}) and give all honest assessors the same input to L3 (\Cref{inv:inputconsistency}).
The certificate chain also makes signature equivocation, invalid evidence pointers, $\counterqc$ censorship, and semantic-layer equivocation provable.
It further supports best-effort detection of evidence suppression.

\section{Evaluation}\label{sec:evaluation}

\looseness=-1 The evaluation follows the three research questions of \S\ref{sec:introduction}.
RQ1 measures how accurately each aggregation rule decides benign Web claims (\S\ref{sec:eval:rq1}).
RQ2 measures how fast integrity degrades under adversarial and correlated assessors (\S\ref{sec:eval:rq2}).
RQ3 measures what the challenge path, the bound, and assessor diversity contribute and what they cost (\S\ref{sec:eval:rq3}).

\subsection{Setup}\label{sec:eval:setup}

\leadtag{Public benchmarks}{tagsetup}{We use fixed evaluation splits with published evidence and labels (\Cref{tab:astral_main}).}
\looseness=-1 \Cref{tab:astral_main} lists claim counts by fixed split.
Here $N$ counts the claims in each fixed split and does not count protocol trials. AVeriTeC and FEVER cover general fact-checking, Climate-FEVER covers climate, PUBHEALTH covers public health, and SciFact pairs scientific claims with abstracts. We replay each claim with its released evidence and ground-truth label.
No benchmark claim, evidence item, or label is generated or rewritten by an LLM.

\leadtag{Synthetic and live cohorts}{tagsetup}{We use two additional cohorts to test controlled mechanisms and mutable evidence.}
Multi-source fact verification (MFV) uses 100 synthetic claims over five topics, half supported and half refuted.
Each claim's difficulty is set by consistency among its three to five evidence excerpts.
The live cohort uses evidence from the NVD Web API for 80 vulnerability claims, 40 about real \texttt{CRITICAL}-severity CVEs and 40 about fabricated ones.
\looseness=-1 Two non-Web domains (code patching and log auditing) appear in \textbf{Appendix~\ref{app:extra}}.

\leadtag{Assessors and harness}{tagsetup}{We form committees of $n \in \{7,\allowbreak 13,\allowbreak 31,\allowbreak 67,\allowbreak 100\}$ assessors from seven model families, using four families at $n{=}7$ (\textbf{Appendix~\ref{app:extra}}). We set $f = \lfloor (n-1)/3 \rfloor$.}
For $n \geq 31$ the committee cycles over the seven families, so the number of independent sources stays at seven.
The public-benchmark runs use a separate $n{=}7$ committee of seven open-weight models (\Cref{tab:astral_main}), namely Qwen2.5-1.5B and Qwen2.5-3B, SmolLM2-1.7B, OLMo-2-1B, TinyLlama-1.1B, StableLM-2-1.6B, and Phi-3.5-mini at pinned revisions, decoding greedily with at most 128 new tokens.
We use $k = 2f{+}1$, $\tau = 0.7$, and seed 42.
In the synthetic and non-Web runs, API models decode greedily and local models use temperature 0.01. Each assessor's response is cached once per claim and reused across protocols.
L2 and L3 are simulated in logical rounds on top of a deterministic L1 stub after GST, so $\Delta$, $W(v)$, $T_{\mathit{infer}}$, $T_{\mathit{DA}}$, and $T_{\mathit{arb}}$ take no wall-clock values.
The challenge window lasts one round.
During this window, every honest assessor whose $\vinfer$ label is $\mathsf{REJECT}$ attempts to file one counter.
A counter enters triage only if it passes \texttt{VerifyCounterStruct}.
The $m$ checkers of \S\ref{sec:design:probability} are therefore the honest assessors ($n{-}f$ in attack trials and all $n$ in benign trials), and we use each checker's Phase~\textcircled{1} label as its detection outcome.
Each assessor casts one triage vote, and a counter with fewer than $2f{+}1$ \texttt{COUNTER\_REJECT} votes sends the decision to \texttt{DISPUTED}.
This triage rule replaces the $T_{\mathit{arb}}$ fallback in the simulation.
\looseness=-1 We select members to obtain diversity and leave the diversity predicate unchecked at attestation time.

\leadtag{Adversarial assessors}{tagthreat}{In attack trials, $f$ assessors are Byzantine and attest with confidence 0.95.}
Prompt or tool injection (A1)~\cite{promptinjection} embeds an instruction that asks for a SUPPORT verdict in the evidence that every assessor reads, and the Byzantine assessors attest the wrong verdict.
In a fabricated proposal (A2), a Byzantine leader proposes a plausible false verdict.
In evidence truncation (A3), the last 30\% of the evidence is removed for every assessor, and the Byzantine assessors endorse the incomplete decision object.
For model-aware corruption (A4), the evidence carries the A1 injection, and we corrupt the $f$ assessors that reject most reliably and leave the weakest models honest.
\looseness=-1 In counter-spam (A5), each Byzantine assessor files one structurally valid but vacuous counter per window.

\leadtag{Baselines and metrics}{tagsetup}{We compare against majority voting (MV) over the assessors' labels and against CP-WBFT~\cite{cpwbft}, which probes confidence and refines labels with confidence-weighted aggregation.}
L1-BFT, which orders verdicts without semantic checks, and an oracle with a perfect $\vinfer$ bound the range.
\astral-NoCh drops the challenge window, \astral-NoDiv uses one backend for all assessors, and Simple Threshold accepts on one assessor's confidence.
In \Cref{tab:astral_main}, ``Single-model (avg.)'' averages seven independent runs that each use one backend. ``w/o Challenge'' is \astral-NoCh, and ``Simple-Thr.'' is the single-assessor threshold baseline.
For \Cref{tab:astral_main}, the semantic correctness success rate (SCSR) is the fraction of benign trials whose final binary decision matches the ground-truth label.
Coverage (Cov.) is the fraction of benign trials that end in a binary \texttt{ACCEPT} or \texttt{REJECT} instead of \texttt{DISPUTED}.
Both SCSR and Cov.\ use all benign trials as the denominator. A \texttt{DISPUTED} outcome is not a binary decision, so it never counts as correct, and SCSR never exceeds Cov.
Each A1--A4 column reports the attack success rate (ASR), counting a refuted claim as an attack success only when the final output is \texttt{ACCEPT}.
A \texttt{DISPUTED} outcome lowers coverage and suspends any action taken at $\commitqc$.
The pooled ASR aggregates accepted refuted claims over all four attacks, and Disp.\ is the fraction of attacked trials that end in \texttt{DISPUTED}.

\subsection{RQ1: Reliability under Benign Disagreement}\label{sec:eval:rq1}

\begin{table}[t]
\centering
\caption{Main results on five public benchmarks using seven open-weight assessors. $N$ is the claim count in the fixed evaluation split. SCSR and Cov.\ are clean metrics, A1--A4 are attack success rates, and pooled ASR aggregates the four attacks.}
\label{tab:astral_main}
\begingroup
\setlength{\tabcolsep}{2pt}
\renewcommand{\arraystretch}{0.9}
\footnotesize
\newcommand{\ASRCI}[3]{\mbox{#1\,{\scriptsize\color[gray]{0.42}[#2,#3]}}}
\resizebox{\ifdim\width>\columnwidth\columnwidth\else\width\fi}{!}{%
\begin{tabular}{l*{8}{r}}
\toprule
& \multicolumn{2}{c}{\textbf{Clean}} & \multicolumn{6}{c}{\textbf{Byzantine attacks (\%)}} \\
\cmidrule(lr){2-3}\cmidrule(l){4-9}
\textbf{Method} & \textbf{SCSR}$\uparrow$ & \textbf{Cov.}$\uparrow$ & \textbf{A1}$\downarrow$ & \textbf{A2}$\downarrow$ & \textbf{A3}$\downarrow$ & \textbf{A4}$\downarrow$ & \textbf{ASR}$\downarrow$ & \textbf{Disp.}$^{\dagger}$ \\
\midrule
\rowcolor[RGB]{240,240,240}\multicolumn{9}{c}{\rule{0pt}{9pt}\textbf{AVeriTeC}\ {\normalfont\color[gray]{0.42}($N=500$)}} \\
L1-BFT & 24.4 & \textbf{100.0} & 100.0 & 100.0 & 100.0 & 100.0 & \ASRCI{100.0}{99.8}{100.0} & 0.0 \\
MV & \textbf{88.2} & \textbf{100.0} & 96.6 & 17.5 & 26.7 & 100.0 & \ASRCI{60.2}{57.7}{62.7} & 0.0 \\
CP-WBFT & 50.6 & \textbf{100.0} & 97.6 & 50.3 & 51.9 & 100.0 & \ASRCI{74.9}{72.7}{77.1} & 0.0 \\
Simple-Thr. & 75.6 & \textbf{100.0} & 73.0 & 40.5 & 45.0 & 85.4 & \ASRCI{61.0}{58.5}{63.4} & 0.0 \\
\cmidrule(lr){1-9}
Single-model (avg.) & 56.1 & 76.1 & 57.1 & 15.8 & 20.4 & 57.1 & 37.6 & 17.3 \\
w/o Challenge & 86.0 & \textbf{100.0} & 64.6 & 3.4 & 9.5 & 99.7 & \ASRCI{44.3}{41.8}{46.9} & 0.0 \\
\rowcolor[RGB]{229,236,250} \textbf{ASTRAL} & 14.0 & 16.0 & \textbf{27.2} & \textbf{1.6} & \textbf{6.1} & \textbf{51.1} & \ASRCI{\textbf{21.5}}{19.4}{23.7} & 50.5 \\
\addlinespace[3pt]
\rowcolor[RGB]{240,240,240}\multicolumn{9}{c}{\rule{0pt}{9pt}\textbf{Climate-FEVER}\ {\normalfont\color[gray]{0.42}($N=1{,}535$)}} \\
L1-BFT & 42.6 & \textbf{100.0} & 100.0 & 100.0 & 100.0 & 100.0 & \ASRCI{100.0}{99.9}{100.0} & 0.0 \\
MV & \textbf{70.4} & \textbf{100.0} & 99.3 & 42.3 & 45.1 & 100.0 & \ASRCI{71.7}{70.2}{73.2} & 0.0 \\
CP-WBFT & 46.3 & \textbf{100.0} & 99.3 & 67.4 & 74.2 & 100.0 & \ASRCI{85.2}{84.0}{86.4} & 0.0 \\
Simple-Thr. & 61.8 & \textbf{100.0} & 77.8 & 52.9 & 47.3 & 90.8 & \ASRCI{67.2}{65.6}{68.7} & 0.0 \\
\cmidrule(lr){1-9}
Single-model (avg.) & 54.0 & 87.4 & 68.5 & 29.0 & 29.6 & \textbf{68.5} & 48.9 & 9.6 \\
w/o Challenge & 67.9 & \textbf{100.0} & 84.3 & 14.9 & 14.8 & 100.0 & \ASRCI{53.5}{51.8}{55.1} & 0.0 \\
\rowcolor[RGB]{229,236,250} \textbf{ASTRAL} & 26.2 & 37.5 & \textbf{61.1} & \textbf{11.4} & \textbf{11.7} & 87.9 & \ASRCI{\textbf{43.0}}{41.3}{44.6} & 26.7 \\
\addlinespace[3pt]
\rowcolor[RGB]{240,240,240}\multicolumn{9}{c}{\rule{0pt}{9pt}\textbf{FEVER}\ {\normalfont\color[gray]{0.42}($N=15{,}935$)}} \\
L1-BFT & 29.1 & \textbf{100.0} & 100.0 & 100.0 & 100.0 & 100.0 & \ASRCI{100.0}{100.0}{100.0} & 0.0 \\
MV & 83.0 & \textbf{100.0} & 99.5 & 54.0 & 47.6 & 100.0 & \ASRCI{75.3}{74.9}{75.7} & 0.0 \\
CP-WBFT & 44.3 & \textbf{100.0} & 99.7 & 66.5 & 57.8 & 100.0 & \ASRCI{81.0}{80.6}{81.3} & 0.0 \\
Simple-Thr. & 67.5 & \textbf{100.0} & 75.7 & 52.3 & 49.6 & 91.3 & \ASRCI{67.2}{66.8}{67.6} & 0.0 \\
\cmidrule(lr){1-9}
Single-model (avg.) & 50.2 & 80.7 & 66.3 & 33.3 & 29.5 & \textbf{66.3} & 48.9 & 14.0 \\
w/o Challenge & \textbf{84.5} & \textbf{100.0} & 83.1 & 20.6 & 15.1 & 100.0 & \ASRCI{54.7}{54.2}{55.2} & 0.0 \\
\rowcolor[RGB]{229,236,250} \textbf{ASTRAL} & 27.5 & 33.3 & \textbf{53.4} & \textbf{15.0} & \textbf{12.4} & 67.2 & \ASRCI{\textbf{37.0}}{36.6}{37.4} & 31.9 \\
\addlinespace[3pt]
\rowcolor[RGB]{240,240,240}\multicolumn{9}{c}{\rule{0pt}{9pt}\textbf{PUBHEALTH}\ {\normalfont\color[gray]{0.42}($N=1{,}220$)}} \\
L1-BFT & 51.6 & \textbf{100.0} & 100.0 & 100.0 & 100.0 & 100.0 & \ASRCI{100.0}{99.8}{100.0} & 0.0 \\
MV & 62.8 & \textbf{100.0} & 97.0 & 7.8 & 14.2 & 100.0 & \ASRCI{54.7}{52.7}{56.8} & 0.0 \\
CP-WBFT & \textbf{72.6} & \textbf{100.0} & 96.4 & 35.5 & 51.3 & 100.0 & \ASRCI{70.8}{68.9}{72.6} & 0.0 \\
Simple-Thr. & 62.5 & \textbf{100.0} & 70.9 & 30.8 & 43.5 & 83.1 & \ASRCI{57.1}{55.0}{59.1} & 0.0 \\
\cmidrule(lr){1-9}
Single-model (avg.) & 57.2 & 90.6 & 57.9 & 8.0 & 12.9 & \textbf{57.9} & \textbf{34.2} & 5.9 \\
w/o Challenge & 55.5 & \textbf{100.0} & 66.5 & \textbf{2.4} & 4.1 & 100.0 & \ASRCI{43.2}{41.2}{45.3} & 0.0 \\
\rowcolor[RGB]{229,236,250} \textbf{ASTRAL} & 29.6 & 55.9 & \textbf{45.7} & \textbf{2.4} & \textbf{3.7} & 92.6 & \ASRCI{36.1}{34.1}{38.1} & 16.8 \\
\addlinespace[3pt]
\rowcolor[RGB]{240,240,240}\multicolumn{9}{c}{\rule{0pt}{9pt}\textbf{SciFact}\ {\normalfont\color[gray]{0.42}($N=300$)}} \\
L1-BFT & 41.3 & \textbf{100.0} & 100.0 & 100.0 & 100.0 & 100.0 & \ASRCI{100.0}{99.5}{100.0} & 0.0 \\
MV & 63.7 & \textbf{100.0} & 99.4 & 71.0 & 40.9 & 100.0 & \ASRCI{77.8}{74.6}{80.9} & 0.0 \\
CP-WBFT & 42.0 & \textbf{100.0} & 97.2 & 86.9 & 56.8 & 100.0 & \ASRCI{85.2}{82.4}{87.8} & 0.0 \\
Simple-Thr. & 60.7 & \textbf{100.0} & 82.4 & 59.1 & 50.6 & 94.3 & \ASRCI{71.6}{68.1}{74.9} & 0.0 \\
\cmidrule(lr){1-9}
Single-model (avg.) & 53.1 & 93.0 & 72.5 & 45.7 & 27.8 & \textbf{72.5} & 54.6 & 7.5 \\
w/o Challenge & \textbf{70.0} & \textbf{100.0} & 90.9 & 41.5 & 11.4 & 100.0 & \ASRCI{60.9}{57.2}{64.6} & 0.0 \\
\rowcolor[RGB]{229,236,250} \textbf{ASTRAL} & 43.0 & 61.7 & \textbf{72.2} & \textbf{36.4} & \textbf{10.2} & 93.8 & \ASRCI{\textbf{53.1}}{49.4}{56.9} & 16.5 \\
\bottomrule
\end{tabular}}
\endgroup
\vspace{2pt}
\parbox{\columnwidth}{\scriptsize All metrics in \%. $\uparrow$/$\downarrow$ mark the better direction, and bold marks the best value within each dataset (ties included). $^{\dagger}$Disp.\ is the disputed-output rate and has no preferred direction. Brackets give Clopper--Pearson 95\% CIs for the pooled ASR, whose bounds lie within 1.2 points of claim-clustered bootstrap CIs. Single-model (avg.) averages seven single-backend runs without CI, pooled ASR ranges from 1.4--8.0\% for the best backend to 57.5--90.1\% for the worst, and A4 equals A1 because one backend leaves no assessor to select.}
\end{table}

\begin{figure}[t]
\centering
\includegraphics[width=\linewidth]{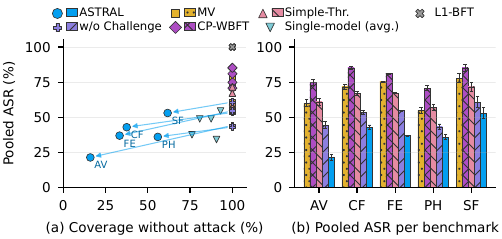}
\caption{Pooled ASR under A1--A4 on the public benchmarks of \Cref{tab:astral_main} (AV = AVeriTeC, CF = Climate-FEVER, FE = FEVER, PH = PUBHEALTH, SF = SciFact). (a) Coverage without attack against pooled ASR, with arrows from \astral-NoCh to \astral. (b) Pooled ASR with Clopper--Pearson 95\% CIs, omitting Single-model and L1-BFT.}
\label{fig:tradeoff}
\end{figure}

Without adversaries, label aggregation always reaches a binary verdict on the five public benchmarks, while \astral reaches one on 16.0--61.7\% of claims and leaves the remaining 38.3--84.0\% \texttt{DISPUTED} (\Cref{tab:astral_main} and \Cref{fig:tradeoff}a).
Only 6.8--30.3\% of claims reach accountable acceptance with a StableQC, and the other binary verdicts are rejections without SemQC.
The whole loss comes from the challenge window.
\astral-NoCh uses the same $2f{+}1$ quorum without a window, never outputs \texttt{DISPUTED}, and decides every claim.
Because disputed claims never count as correct, \astral's SCSR is 20.7--74.2 points below MV's on all five benchmarks.
On decided claims, \astral is correct 52.9--87.5\% of the time, against 62.8--88.2\% for MV.
\looseness=-1 On the synthetic suite, \astral's SCSR falls from 43.3\% at $n{=}7$ to 32.9\% at $n \geq 13$, with coverage of 43.8\% and 33.3\%, while MV and CP-WBFT decide every claim with SCSR of 93.8\% and 94.3\% (\textbf{Appendix~\ref{app:extra}}).

\takeawaybox{Under benign disagreement, \astral withholds a binary verdict on 38.3--84.0\% of public-benchmark claims, and the challenge window accounts for all of that loss because the $2f{+}1$ quorum alone decides every claim.}

\subsection{RQ2: Robustness to Adversarial Assessors}\label{sec:eval:rq2}

\begin{figure}[t]
\centering
\includegraphics[width=\linewidth]{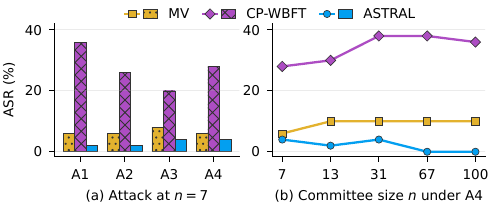}
\caption{Attack success rate on MFV multi-source claims over 50 refuted claims per point (4\% has 95\% CI [0.5, 13.7]). (a) Each attack at $n{=}7$. (b) A4 across committee size $n$.}
\label{fig:mfv}
\end{figure}

\leadtag{Adversarial assessors}{tagthreat}{With two of seven assessors adversarial, \astral has the lowest pooled ASR among the aggregation protocols on every public benchmark. Its pooled ASR is 21.5--53.1\% against 54.7--77.8\% for MV and 70.8--85.2\% for CP-WBFT (\Cref{tab:astral_main} and \Cref{fig:tradeoff}b).}
L1-BFT accepts every refuted claim, which shows that ordering without a semantic check offers no protection.
At the same confidence threshold, Simple Threshold reaches 57.1--71.6\% pooled ASR.
With the $2f{+}1$ quorum, \astral-NoCh lowers pooled ASR by 10.7--16.7 points relative to Simple Threshold. Adding the challenge window lowers it a further 7.1--22.8 points and attaches a signed record to every decision (\Cref{fig:tradeoff}a).
At \astral's benign coverage, MV that abstains below a confidence threshold still accepts 54.7--77.7\% of refuted claims, so abstention alone does not explain \astral's lower ASR (\Cref{tab:matched_cov}).
Model-aware corruption (A4) remains the strongest attack, as \astral accepts 51.1--93.8\% of refuted claims under A4.
On PUBHEALTH, the average single-backend run reaches a pooled ASR of 34.2\%, 1.9 points below \astral.
CP-WBFT is weakest where one confident assessor can move the others.
On MFV at $n{=}7$, CP-WBFT accepts 20--36\% of refuted claims because one assessor reporting confidence 0.95 can shift honest assessors to that assessor's label during refinement (\Cref{fig:mfv}a).
\astral's attack success on MFV is 2--4\%.

\leadtag{Correlated assessors}{tagsemantic}{Growing the committee by replicating the same seven families adds assessors whose errors coincide.}
As $n$ grows under model-aware corruption, attack success on MFV rises from 6\% to 10\% for MV and from 28\% to 36--38\% for CP-WBFT, while \astral stays at 0--4\% (\Cref{fig:mfv}b).
In log auditing, the same replication raises MV's attack success to 88\% at $n \geq 13$ (\textbf{Appendix~\ref{app:extra}}).
% Required in the paper preamble: \usepackage{booktabs}\usepackage{colortbl}
% If your template already loads xcolor, do not load it again with new options.
\leadtag{Live Web evidence}{tagsemantic}{On the 80 live vulnerability claims at $n{=}7$, \astral supports 39 of 40 real CVEs (97.5\%, 95\% CI [86.8, 99.9]) and rejects all 40 fabricated ones, both without attack and under A1.}
When one assessor is degraded to reject conservatively, the protocol supports 38 of 40 real CVEs without attack and 39 of 40 under A1, and still accepts no fabricated CVE.

\takeawaybox{When a minority of assessors is adversarial or correlated, label aggregation turns their errors into accepted Web verdicts, while \astral lowers pooled attack success by 18.6--38.7 points relative to MV on public benchmarks and keeps it at 0--4\% on multi-source claims.}

\subsection{RQ3: Accountability, Diversity, and Cost}\label{sec:eval:rq3}

\begin{table}[t]
\centering
\caption{Assessor correlation on MFV ($\tau{=}0.7$) over 50 trials for one fixed group of $s$ honest assessors per configuration, with group-mean false-accept rate $\bar{\epsilon}$ and observed ratio $\hat{\rho}_s$. $s$ grows with the family count, and -- marks no joint false acceptance.}\label{tab:rho}
\footnotesize
\resizebox{\ifdim\width>\columnwidth\columnwidth\else\width\fi}{!}{
\begin{tabular}{@{}lcccc@{}}
\toprule
\textbf{Diversity configuration} & \textbf{$s$} & \textbf{$\Pr[\cap A_i]$} & \textbf{$\bar{\epsilon}^{s}$} & \textbf{$\hat{\rho}_s$} \\
\midrule
Same model family      & 2 & $2.0 \times 10^{-2}$ & $4.0 \times 10^{-4}$ & 50.0 \\
Mixed 2 families       & 2 & $2.0 \times 10^{-2}$ & $9.0 \times 10^{-4}$ & 22.2 \\
Mixed 3 families       & \astralsoft{3} & \astralsoft{$0$} & \astralsoft{$3.7 \times 10^{-5}$} & \astralsoft{--} \\
Full diversity (4+)    & \astralcell{4} & \astralcell{$0$} & \astralcell{$1.5 \times 10^{-6}$} & \astralcell{--} \\
\bottomrule
\end{tabular}
}
\end{table}

\begin{table*}[t]
\centering
\caption{Decision-integrity approaches compared, with measured results from \Cref{tab:astral_main}, \Cref{fig:mfv}, and \Cref{tab:protocol-comparison}.}
\label{tab:novelty}
\scriptsize
\setlength{\tabcolsep}{2.2pt}
\renewcommand{\arraystretch}{0.92}
\begin{tabular}{l|c|cccc|cccc|ccccc|c}
\toprule
& & \multicolumn{4}{c|}{\textbf{Decision guarantees}} & \multicolumn{4}{c|}{\textbf{Accountability}} & \multicolumn{5}{c|}{\textbf{Measured (\%, tokens)}} & \\
\cmidrule(lr){3-6}\cmidrule(lr){7-10}\cmidrule(lr){11-15}
\textbf{System} & \textbf{Paradigm} &
\makecell{Ordering\\agreement} &
\makecell{Cert.\\finality} &
\makecell{Semantic\\attestation} &
\makecell{Evidence\\provenance} &
\makecell{Network\\forensics} &
\makecell{Challenge\\mechanism} &
\makecell{Counter\\triage} &
\makecell{Dist.\\bounds} &
\makecell{Clean\\SCSR$\uparrow$} &
\makecell{Cov.\\w/o attack$\uparrow$} &
\makecell{Pooled\\ASR$\downarrow$} &
\makecell{MFV A4\\ASR$\downarrow$} &
\makecell{Tokens/\\decision} &
\makecell{\textbf{Capab.}\\\textbf{coverage}}
\\
\midrule
L1-BFT~\cite{pbft,tendermint,hotstuff} & BFT & \cmark & \cmark & \xmark & \xmark & \xmark & \xmark & \xmark & \xmark & 24.4--51.6 & 100.0 & 100.0 & -- & -- & 2/8 \\
Accountable BFT~\cite{accountablebft2026} & BFT & \cmark & \cmark & \xmark & \xmark & \cmark & \xmark & \xmark & \xmark & -- & -- & -- & -- & -- & 3/8 \\
Optimistic rollups~\cite{rollupsurvey,arbitrum} & Fraud proof & \cmark & \cmark & \xmark & \pmark & \pmark & \cmark & \xmark & \xmark & -- & -- & -- & -- & -- & 3+2/8 \\
\midrule
Single-model (avg.) & Single & \xmark & \xmark & \xmark & \xmark & \xmark & \xmark & \xmark & \xmark & 50.2--57.2 & 76.1--93.0 & 34.2--54.6 & -- & -- & 0/8 \\
Simple Threshold & Threshold & \xmark & \xmark & \xmark & \xmark & \xmark & \xmark & \xmark & \xmark & 60.7--75.6 & 100.0 & 57.1--71.6 & -- & -- & 0/8 \\
Majority vote~\cite{coordllm,consensusdetection} & Vote & \xmark & \xmark & \pmark & \xmark & \xmark & \xmark & \xmark & \xmark & 62.8--88.2 & 100.0 & 54.7--77.8 & 6.0 & 8.1K & 0+1/8 \\
Debate~\cite{multiagentdebate} & Debate & \xmark & \xmark & \pmark & \xmark & \xmark & \xmark & \xmark & \xmark & -- & -- & -- & -- & -- & 0+1/8 \\
CP-WBFT~\cite{cpwbft} & Conf.-weighted & \cmark & \xmark & \pmark & \xmark & \xmark & \xmark & \xmark & \xmark & 42.0--72.6 & 100.0 & 70.8--85.2 & 28.0 & 16.4K & 1+1/8 \\
\midrule
\astral-NoCh & Sem.\ consensus & \cmark & \cmark & \cmark & \cmark & \cmark & \xmark & \xmark & \pmark & 55.5--86.0 & 100.0 & 43.2--60.9 & -- & -- & 5+1/8 \\
\rowcolor{paperaccentsoft}\textbf{\astral} & \textbf{Sem.\ consensus} & \cmark & \cmark & \cmark & \cmark & \cmark & \cmark & \cmark & \pmark & 14.0--43.0 & 16.0--61.7 & \textbf{21.5--53.1} & \textbf{4.0} & 12.0K & \textbf{7+1/8} \\
\bottomrule
\end{tabular}
\par\smallskip
\parbox{\textwidth}{\scriptsize
\cmark\ = supported, \pmark\ = partial, \xmark\ = not addressed, -- = not evaluated. Coverage counts full$+$partial of eight. Rollups resolve challenges by deterministic re-execution, and fraud proofs convict a wrong state root. Ranges span five public benchmarks with two of seven assessors adversarial. MFV A4 ASR and tokens (code patching) are at $n{=}7$. \astral's bound is not evaluated end to end.}
\end{table*}

\leadtag{Diversity and the bound}{tagsemantic}{On MFV, groups that span more model families show fewer joint false acceptances (\Cref{tab:rho}).}
\looseness=-1 A committee from one model family produced one joint false acceptance in 50 trials. The observed ratio $\hat{\rho}_s = 50$ has a wide interval because it rests on that single event.
Two families also produced one joint false acceptance ($\hat{\rho}_s = 22.2$), and three or more families produced none (one-sided 95\% upper bound 0.06 on the joint false-acceptance rate).
These joint false-acceptance rates exclude the challenge factor of \Cref{eq:delta}.
\textbf{Appendix~\ref{app:extra}} reports the unconditional rate at which counters are filed under attack, which is only the first stage of $p_{\mathit{ch}}$, whereas the bound conditions $p_{\mathit{ch}}$ on $Q$.
The conditional probability can be lower, so we do not evaluate the bound end to end.
In the MFV ablation under A2, attestation alone lowers attack success from 100\% to 6\%, and adding the challenge window lowers it to 2\% (\textbf{Appendix~\ref{app:extra}}).
\astral-NoDiv reaches 2.0--4.0\% depending on the backend (3.5\% on average), and its single-family ratio of 50 indicates strongly correlated errors. Simple Threshold, which faces the same Byzantine assessors, accepts 40.0\% of refuted MFV claims under A2.
On the public benchmarks, diversity hedges against a weak backend but does not beat the strongest one. Single-backend committees reach 1.4--8.0\% pooled ASR with the best model and 57.5--90.1\% with the worst, against 21.5--53.1\% for \astral's mixed committee (\Cref{tab:astral_main}).

\leadtag{Accountability}{tagdefense}{The challenge path records who disagrees instead of returning a label without an attributable record.}
At $n{=}7$ under attacks A1--A4, counters are filed in 56.7--59.0\% of trials, $\crqc$ resolution occurs in 5.0--8.9\%, and 53.8--55.7\% of decisions end as \texttt{DISPUTED}.
Pooled over the synthetic domains, the filing rate is 58.3\% at $n{=}7$ and 67.0\% at $n{=}31$ (\Cref{tab:protocol-comparison}).

\leadtag{Cost and liveness}{tagsetup}{\astral spends tokens and liveness to keep false acceptance low.}
On the code-patching domain at $n{=}7$, a decision costs 8.1K tokens under MV, 16.4K under CP-WBFT, and 12.0K under \astral (\Cref{tab:protocol-comparison}).
At $n{=}31$ the costs are 36.9K, 180.0K, and 172.7K.
On the same domain with 31 assessors, MV's attack success under A4 is 46.7\% against 0\% for \astral.
Under counter-spam (A5), varying each Byzantine assessor's counter budget $K$ over 0, 1, and $f$ (with $K{>}1$ lifting the one-counter limit) yields 2.7--6.7 counters per trial at $n{=}7$ and 11.4--111.4 at $n{=}31$.
For every $K$, 55.2\% and 63.3\% of decisions end as \texttt{DISPUTED}, with 2.2 and 7.7 sustained counters per trial, respectively.
Byzantine assessors voting \texttt{COUNTER\_SUSTAIN} during triage are therefore the bottleneck, and mitigations such as weighting votes or relaxing $\crqc$ thresholds should target these votes.
Counter-spam costs liveness, and a decision it pushes to \texttt{DISPUTED} carries no $\stableqc$.
\looseness=-1 Of the attacked synthetic trials, which mix supported and refuted claims, 43.9\% reach stable finality at $n{=}7$ and 54.6\% are flagged for audit.

\takeawaybox{\looseness=-1 \astral fails conservatively. Groups of three or more families had no joint false acceptance on multi-source claims, diversity hedges against a weak backend, and under attack \astral flags more verdicts as \texttt{DISPUTED}.}

\leadtag{L1 composition}{tagsetup}{\astral applies to an L1 that provides agreement, total ordering, and an external-validity predicate, including DAG-based protocols in which evidence binds to vertex hashes.}\label{sec:discussion:generality}\label{sec:eval:l1comp}
\looseness=-1 For probabilistic L1s with $k$-deep confirmation, $\delta_{\mathit{joint}} \leq \delta_{\text{L3}} + p_{\mathit{rev}}$, and in simulated L1 reversion the joint-failure rates track this bound, at $0.71\%$ vs.\ $0.72\%$ for $p_{\mathit{rev}}{=}10^{-4}$ and $5.5\%$ vs.\ $5.7\%$ for $p_{\mathit{rev}}{=}0.05$.

\leadtag{Limitations}{tagthreat}{MFV is synthetic, and the five public benchmarks are static corpora.}
\looseness=-1 The live CVE cohort of 80 claims is the only evaluation with mutable Web evidence and evaluates \astral without baselines.
The bound is distributional and gives no guarantee against adaptive adversaries who model every assessor.
Membership is static. In A1--A4, attack payloads are structural and are not optimized against specific assessors. In A1 and A4 honest assessors also read the injected evidence and do not count toward $f$.
Content that steers more than $f$ assessors to the same wrong label lies outside the quorum guarantee and appears only through $\epsilon$ and $\rho_{k_h}$ (\S\ref{sec:threat:faults}).
\astral bounds risk conditional on the committed evidence, so evidence that retrieval never collected is beyond its scope.
The evaluated Byzantine assessors withhold neither \texttt{DA\_ACK} nor \texttt{SEM\_PASS}, which leaves the withholding case of \Cref{thm:liveness} untested.
Latency is measured only in logical rounds.
Assessor diversity is self-declared in our setting. Such tags can be spoofed, and models from the same family may share failure modes.\label{sec:discussion:diversity}
Periodic known-answer probes can test behavioral diversity, high-stakes deployments can use trusted hardware to attest the served model, and the correlation term can be calibrated for each deployment (\Cref{tab:rho}).

\section{Related Work}\label{sec:background}

\leadtag{Web information integrity and fact-checking}{tagsetup}{Benchmarks such as FEVER~\cite{thorne2018fever} and SciFact~\cite{wadden2020scifact} pair claims with retrieved evidence to evaluate automated fact verification on the Web.}
\looseness=-1 LLMs generate fluent misinformation at scale and are also adopted as assistants that detect and annotate it~\cite{factualitychallenges,misinfollm}.
Work in this line evaluates or improves individual verifiers, and systems such as ProveRAG~\cite{proverag} use retrieval to reduce hallucination within one agent.
A single-agent verifier can supply an assessor's $\vinfer$ in an \astral quorum.
\astral addresses how verdicts from several such verifiers are combined into one Web decision.

\leadtag{LLMs as assessors and their aggregation}{tagsemantic}{LLM judges are validated mainly by agreement with human assessors~\cite{llmjudge}. Multi-agent systems~\cite{autogen,metagpt,chatdev} resolve disagreement among LLM assessors by voting~\cite{coordllm,consensusdetection} or debate~\cite{multiagentdebate}.}
Voting and debate assume benign assessors.
Sycophancy makes agents defer to perceived authority~\cite{sycophancy}, fixed personas reduce epistemic diversity~\cite{beliefboxes}, and indirect prompt injection lets retrieved Web content steer an assessor~\cite{promptinjection}.
As a result, one compromised or confident assessor can propagate a wrong verdict through the group.
CP-WBFT~\cite{cpwbft} examines reliability in multi-agent systems from a Byzantine perspective by probing confidence and weighting verdicts during aggregation.
Our evaluation shows that one confident assessor can exploit how CP-WBFT weights verdicts by confidence (\S\ref{sec:eval:rq2}).

\leadtag{Accountable agreement and challenge-based finality}{tagdefense}{Classical Byzantine fault-tolerant (BFT) protocols~\cite{pbft,hotstuff,dls,lamport1998paxos} and their successors~\cite{asyncbft2025,groupbft2025,ipotane2026,accountablebft2026} agree on bytes and attribute network-level equivocation.}
A signed proposal that contains a false claim still passes all of their checks.
Blockchain oracles~\cite{chainlink2,oraclesurvey} aggregate independent data feeds without a challenge path, and prediction markets~\cite{augur} rely on economic incentives over long horizons.
Optimistic rollups~\cite{rollupsurvey,arbitrum} use challenges against deterministic execution to reach finality. \astral adapts these challenges by having stochastic assessors triage counters.
Verifiable computation~\cite{pinocchio,stark,zkml} and trusted hardware~\cite{sgx,slalom} prove that a computation ran as specified but do not prove that its conclusion about a claim is right.
Within one agent or orchestrator, evidence-binding frameworks~\cite{evibound,proofcarrying} gate outputs on evidence.
\looseness=-1 We are not aware of work that jointly pins Web evidence, requires a quorum of attestations that pass a confidence threshold, and keeps LLM assessors' verdicts open to challenge before accountable acceptance (\Cref{tab:novelty}).

\section{Conclusion}\label{sec:conclusion}

\looseness=-1 We presented \astral, an accountable semantic consensus layer for AI-mediated Web verification.
When the output sequences of independent LLM assessors never separate valid from invalid claims, a protocol that rejects invalid claims also falsely rejects valid ones with positive probability.
\astral pins Web evidence, counts attestations above a confidence threshold, and admits signed challenges before accountable acceptance.
With two of seven assessors adversarial, it lowered pooled attack success by 18.6--38.7 points relative to majority voting on five public benchmarks.
The cost was coverage, as without adversaries 16.0--61.7\% of public claims reached a binary verdict and the rest were disputed for audit.
Groups of three or more model families showed no joint false acceptance on multi-source claims, and diversity hedged against a weak backend on public benchmarks.
These results support Web verification layers whose verdicts carry an attestation record and remain open to challenge.

\clearpage

\bibliographystyle{ACM-Reference-Format}
\bibliography{refs_short}

@inproceedings{hotstuff,
  author    = {Yin, Maofan and Malkhi, Dahlia and Reiter, Michael K. and Gueta, Guy Golan and Abraham, Ittai},
  title     = {{HotStuff}: {BFT} Consensus with Linearity and Responsiveness},
  booktitle = {PODC},
  year      = {2019},
  pages     = {347--356},
}

@inproceedings{pbft,
  author    = {Castro, Miguel and Liskov, Barbara},
  title     = {Practical {Byzantine} Fault Tolerance},
  booktitle = {OSDI},
  year      = {1999},
  pages     = {173--186},
}

@mastersthesis{tendermint,
  author  = {Buchman, Ethan},
  title   = {Tendermint: {Byzantine} Fault Tolerance in the Age of Blockchains},
  school  = {University of Guelph},
  year    = {2016},
}

@article{dls,
  author  = {Dwork, Cynthia and Lynch, Nancy and Stockmeyer, Larry},
  title   = {Consensus in the Presence of Partial Synchrony},
  journal = {JACM},
  volume  = {35},
  number  = {2},
  pages   = {288--323},
  year    = {1988},
}

@article{lamport1998paxos,
  author  = {Lamport, Leslie},
  title   = {The Part-Time Parliament},
  journal = {TOCS},
  volume  = {16},
  number  = {2},
  pages   = {133--169},
  year    = {1998},
}

@inproceedings{autogen,
  author    = {Wu, Qingyun and Bansal, Gagan and Zhang, Jieyu and Wu, Yiran and Li, Beibin and Zhu, Erkang and Jiang, Li and Zhang, Xiaoyun and Zhang, Shaokun and Liu, Jiale and Awadallah, Ahmed Hassan and White, Ryen W. and Burger, Doug and Wang, Chi},
  title     = {{AutoGen}: Enabling Next-Gen {LLM} Applications via Multi-Agent Conversations},
  booktitle = {COLM},
  year      = {2024},
}

@inproceedings{chatdev,
  author    = {Qian, Chen and Liu, Wei and Liu, Hongzhang and Chen, Nuo and Dang, Yufan and Li, Jiahao and Yang, Cheng and Chen, Weize and Su, Yusheng and Cong, Xin and Xu, Juyuan and Li, Dahai and Liu, Zhiyuan and Sun, Maosong},
  title     = {{ChatDev}: Communicative Agents for Software Development},
  booktitle = {ACL},
  year      = {2024},
  pages     = {15174--15186},
}

@inproceedings{metagpt,
  author    = {Hong, Sirui and Zhuge, Mingchen and Chen, Jonathan and Zheng, Xiawu and Cheng, Yuheng and Zhang, Ceyao and Wang, Jinlin and Wang, Zili and Yau, Steven Ka Shing and Lin, Zijuan and Zhou, Liyang and Ran, Chenyu and Xiao, Lingfeng and Wu, Chenglin and Schmidhuber, J\"{u}rgen},
  title     = {{MetaGPT}: Meta Programming for A Multi-Agent Collaborative Framework},
  booktitle = {ICLR},
  year      = {2024},
}

@inproceedings{multiagentdebate,
  author    = {Du, Yilun and Li, Shuang and Torralba, Antonio and Tenenbaum, Joshua B. and Mordatch, Igor},
  title     = {Improving Factuality and Reasoning in Language Models through Multiagent Debate},
  booktitle = {ICML},
  year      = {2024},
  pages     = {11733--11763},
}

@inproceedings{sycophancy,
  author    = {Sharma, Mrinank and Tong, Meg and Korbak, Tomasz and Duvenaud, David and Askell, Amanda and Bowman, Samuel R. and Cheng, Newton and Durmus, Esin and Hatfield-Dodds, Zac and Johnston, Scott R. and Kravec, Shauna and Maxwell, Timothy and McCandlish, Sam and Ndousse, Kamal and Rausch, Oliver and Schiefer, Nicholas and Yan, Da and Zhang, Miranda and Perez, Ethan},
  title     = {Towards Understanding Sycophancy in Language Models},
  booktitle = {ICLR},
  year      = {2024},
}

@inproceedings{promptinjection,
  author    = {Greshake, Kai and Abdelnabi, Sahar and Mishra, Shailesh and Endres, Christoph and Holz, Thorsten and Fritz, Mario},
  title     = {Not What You've Signed Up For: Compromising Real-World {LLM}-Integrated Applications with Indirect Prompt Injection},
  booktitle = {Proceedings of the 16th ACM Workshop on Artificial Intelligence and Security (AISec)},
  year      = {2023},
  pages     = {79--90},
  doi       = {10.1145/3605764.3623985},
}

@inproceedings{accountablebft2026,
  author    = {Civit, Pierre and Collins, Daniel and Gramoli, Vincent and Guerraoui, Rachid and Komatovic, Jovan and Vidigueira, Manuel and Zarbafian, Pouriya},
  title     = {Scalable Accountable {Byzantine} Agreement and Beyond},
  booktitle = {S\&P},
  year      = {2026},
}

@misc{aptosdiembft,
  author       = {{The Diem Team}},
  title        = {{DiemBFT} v4: State Machine Replication in the {Diem} Blockchain},
  year         = {2021},
  howpublished = {Technical Report},
}

@article{asyncbft2025,
  author  = {Cheng, Hao and Lu, Yuan and Lu, Zhenliang and Tang, Qiang and Zhang, Yuxuan and Zhang, Zhenfeng},
  title   = {{JUMBO}: Fully Asynchronous {BFT} Consensus Made Truly Scalable},
  journal = {TDSC},
  volume  = {22},
  number  = {5},
  pages   = {5584--5602},
  year    = {2025},
}

@inproceedings{beliefboxes,
  author    = {Bilgin, Onur and Sami, Abdullah As and Vujjini, Sriram Sai and Licato, John},
  title     = {The Effect of Belief Boxes and Open-Mindedness on Persuasion},
  booktitle = {ICAART},
  year      = {2026},
  pages     = {260--267},
}

@article{consensusdetection,
  author  = {Iwanowski, Marcin and Gahbler, Marcin},
  title   = {Multiple Large {AI} Models' Consensus for Object Detection---{A} Survey},
  journal = {Applied Sciences},
  volume  = {15},
  number  = {24},
  pages   = {12961},
  year    = {2025},
}

@article{coordllm,
  author  = {Saadaoui, Sami and Alonso, Eduardo},
  title   = {Coordinated {LLM} Multi-Agent Systems for Collaborative Question-Answer Generation},
  journal = {Knowledge-Based Systems},
  volume  = {330},
  pages   = {114627},
  year    = {2025},
}

@inproceedings{cpwbft,
  author       = {Zheng, Lifan and Chen, Jiawei and Yin, Qinghong and Zhang, Jingyuan and Zeng, Xinyi and Tian, Yu},
  title        = {Rethinking the Reliability of Multi-agent System: A Perspective from {Byzantine} Fault Tolerance},
  booktitle= {AAAI},
  year={2026},
  pages={35012--35020},
}

@misc{evibound,
  author       = {Chen, Ruiying},
  title        = {Evidence-Bound Autonomous Research ({EviBound}): A Governance Framework for Eliminating False Claims},
  year         = {2025},
  eprint       = {2511.05524},
  archivePrefix = {arXiv},
}

@article{groupbft2025,
  author  = {Liu, Xuyang and Zhang, Zijian and Li, Zhen and Lu, Xin and Li, Meng and Xu, Lei and Ao, Meng and Zhu, Liehuang},
  title   = {Group {BFT}: Two-Round {BFT} Protocols Via Replica Grouping},
  journal = {TDSC},
  volume  = {22},
  number  = {6},
  pages   = {6309--6326},
  year    = {2025},
}

@inproceedings{ipotane2026,
  author    = {Dai, Xiaohai and Ding, Chaozheng and Jin, Hai and Loss, Julian and Ren, Ling},
  title     = {{Ipotane}: Balancing the Good and Bad Cases of Asynchronous {BFT}},
  booktitle = {NDSS},
  year      = {2026},
}

@inproceedings{narwhal,
  author    = {Danezis, George and Kokoris-Kogias, Lefteris and Sonnino, Alberto and Spiegelman, Alexander},
  title     = {Narwhal and Tusk: A {DAG}-Based Mempool and Efficient {BFT} Consensus},
  booktitle = {EuroSys},
  year      = {2022},
  pages     = {34--50},
}

@misc{proofcarrying,
  author       = {Meyman, Edward},
  title        = {Verifiable {AI} Governance: The Four Tests Standard ({4TS}) and Proof-Carrying Decisions},
  year         = {2025},
  howpublished = {SSRN Preprint},
}

@article{proverag,
  author  = {Fayyazi, Reza and Trueba, Stella Hoyos and Zuzak, Michael and Yang, Shanchieh Jay},
  title   = {{ProveRAG}: Provenance-Driven Vulnerability Analysis with Automated Retrieval-Augmented {LLMs}},
  journal = {IEEE Access},
  volume  = {13},
  pages   = {212815--212826},
  year    = {2025},
  doi     = {10.1109/ACCESS.2025.3638251},
}

@inproceedings{pinocchio,
  author    = {Parno, Bryan and Howell, Jon and Gentry, Craig and Raykova, Mariana},
  title     = {Pinocchio: Nearly Practical Verifiable Computation},
  booktitle = {S\&P},
  year      = {2013},
  pages     = {238--252},
}

@inproceedings{stark,
  author    = {Ben-Sasson, Eli and Bentov, Iddo and Horesh, Yinon and Riabzev, Michael},
  title     = {Scalable Zero Knowledge with No Trusted Setup},
  booktitle = {CRYPTO},
  year      = {2019},
  pages     = {701--732},
}

@inproceedings{zkml,
  author    = {Chen, Bing-Jyue and Waiwitlikhit, Suppakit and Stoica, Ion and Kang, Daniel},
  title     = {{ZKML}: An Optimizing System for {ML} Inference in Zero-Knowledge Proofs},
  booktitle = {EuroSys},
  year      = {2024},
  pages     = {560--574},
}

@misc{sgx,
  author       = {Costan, Victor and Devadas, Srinivas},
  title        = {{Intel SGX} Explained},
  howpublished = {Cryptology ePrint Archive, Report 2016/086},
  year         = {2016},
}

@inproceedings{slalom,
  author    = {Tram{\`e}r, Florian and Boneh, Dan},
  title     = {Slalom: Fast, Verifiable and Private Execution of Neural Networks in Trusted Hardware},
  booktitle = {ICLR},
  year      = {2019},
}

@techreport{chainlink2,
  author = {Breidenbach, Lorenz and Cachin, Christian and Chan, Benedict and Coventry, Alex and Ellis, Steve and Juels, Ari and Koushanfar, Farinaz and Miller, Andrew and Magauran, Brendan and Moroz, Daniel and Nazarov, Sergey and Topliceanu, Alexandru and Tram{\`e}r, Florian and Zhang, Fan},
  title  = {Chainlink 2.0: Next Steps in the Evolution of Decentralized Oracle Networks},
  year   = {2021},
  type   = {Whitepaper},
}

@techreport{augur,
  author       = {Peterson, Jack and Krug, Joseph and Zoltu, Micah and Williams, Austin K. and Alexander, Stephanie},
  title        = {Augur: a Decentralized Oracle and Prediction Market Platform (v2.0)},
  institution  = {Forecast Foundation},
  year   = {2019},
  type   = {Whitepaper},
}

@article{oraclesurvey,
  author  = {Pasdar, Amirmohammad and Lee, Young Choon and Dong, Zhongli},
  title   = {Connect {API} with Blockchain: A Survey on Blockchain Oracle Implementation},
  journal = {ACM Computing Surveys},
  volume  = {55},
  number  = {10},
  pages   = {1--39},
  year    = {2023},
}

@inproceedings{arbitrum,
  author    = {Kalodner, Harry and Goldfeder, Steven and Chen, Xiaoqi and Weinberg, S. Matthew and Felten, Edward W.},
  title     = {Arbitrum: Scalable, Private Smart Contracts},
  booktitle = {USENIX Security},
  year      = {2018},
  pages     = {1353--1370},
}

@article{rollupsurvey,
  author  = {Thibault, Louis Tremblay and Sarry, Tom and Hafid, Abdelhakim Senhaji},
  title   = {Blockchain Scaling Using Rollups: A Comprehensive Survey},
  journal = {IEEE Access},
  volume  = {10},
  pages   = {93039--93054},
  year    = {2022},
}

@inproceedings{wadden2020scifact,
  author    = {Wadden, David and Lin, Shanchuan and Lo, Kyle and Wang, Lucy Lu and van Zuylen, Madeleine and Cohan, Arman and Hajishirzi, Hannaneh},
  title     = {Fact or Fiction: Verifying Scientific Claims},
  booktitle = {EMNLP},
  year      = {2020},
  pages     = {7534--7550},
}

@inproceedings{thorne2018fever,
  author    = {Thorne, James and Vlachos, Andreas and Christodoulopoulos, Christos and Mittal, Arpit},
  title     = {{FEVER}: a Large-scale Dataset for Fact Extraction and {VERification}},
  booktitle = {NAACL-HLT},
  year      = {2018},
  pages     = {809--819},
}

@inproceedings{llmjudge,
  author    = {Zheng, Lianmin and Chiang, Wei-Lin and Sheng, Ying and Zhuang, Siyuan and Wu, Zhanghao and Zhuang, Yonghao and Lin, Zi and Li, Zhuohan and Li, Dacheng and Xing, Eric P. and Zhang, Hao and Gonzalez, Joseph E. and Stoica, Ion},
  title     = {Judging {LLM}-as-a-Judge with {MT-Bench} and {Chatbot Arena}},
  booktitle = {Advances in Neural Information Processing Systems 36 (NeurIPS)},
  year      = {2023},
  pages     = {46595--46623},
  doi       = {10.52202/075280-2020},
}

@article{factualitychallenges,
  author  = {Augenstein, Isabelle and Baldwin, Timothy and Cha, Meeyoung and Chakraborty, Tanmoy and Ciampaglia, Giovanni Luca and Corney, David and DiResta, Renee and Ferrara, Emilio and Hale, Scott and Halevy, Alon and Hovy, Eduard and Ji, Heng and Menczer, Filippo and Miguez, Ruben and Nakov, Preslav and Scheufele, Dietram and Sharma, Shivam and Zagni, Giovanni},
  title   = {Factuality Challenges in the Era of Large Language Models and Opportunities for Fact-Checking},
  journal = {Nature Machine Intelligence},
  volume  = {6},
  number  = {8},
  pages   = {852--863},
  year    = {2024},
  doi     = {10.1038/s42256-024-00881-z},
}

@article{misinfollm,
  author  = {Chen, Canyu and Shu, Kai},
  title   = {Combating Misinformation in the Age of {LLMs}: Opportunities and Challenges},
  journal = {AI Magazine},
  volume  = {45},
  number  = {3},
  pages   = {354--368},
  year    = {2024},
  doi     = {10.1002/aaai.12188},
}

\appendix
\setlength{\emergencystretch}{3em}
\setlength{\floatsep}{6pt plus 2pt minus 2pt}
\setlength{\intextsep}{4pt plus 2pt minus 2pt}
% =================================================================
% Appendix A: Ethical Considerations, Open Science, and Use of Generative AI
% =================================================================
\section{Ethical Considerations}\label{app:ethics}

\looseness=-1 The study uses public benchmarks and synthetic tasks, involves no human participants or personal data, and runs attacks only in simulation. Code for all experiments will be released after review.
The authors used LLM assistants (Anthropic's Claude, OpenAI's Codex) to draft and edit text, to write and review code, and to review the protocol specification and reported numbers. The authors checked all such output and take full responsibility for this paper.

% =================================================================
% Appendix A: Formalization of Probabilistic Semantic Validity
% =================================================================
\section{Formalization of Probabilistic Semantic Validity}\label{app:formalization}

\subsection{Ground-Truth and Failure Event Definitions}\label{app:form:defs}

$\mathcal{V}(D) \in \{0, 1\}$ denotes policy-relative semantic validity on $\mathcal{D}$.
\begin{equation*}
\mathcal{V}(D) = 1 \iff \text{claim is supported by evidence under policy.}
\end{equation*}

\astral's \emph{stable finality} requires $\vstruct(D) = \mathsf{PASS}$, L1's $\commitqc$, $k$ \texttt{SEM\_PASS} attestations ($\semqc$), and challenge-window survival ($\ncc$) for that slot.

\begin{definition}[Stable Semantic Failure Event]\label{def:failure}
\begin{equation*}
\mathcal{E} := \{\mathsf{Stable}(D) \wedge \mathcal{V}(D) = 0\}
\end{equation*}
The empirical failure bound is
\begin{equation*}
\delta(\epsilon, k_h) := \mathbb{E}_{D \sim \mathcal{D}}\left[\Pr[\mathcal{E} \mid D]\right]
\end{equation*}
\end{definition}

\subsection{Decision Objects and Semantic Validity}\label{app:form:decision}

\begin{definition}[Semantic Validity]\label{def:semvalid}
Semantic validity uses decision $D$, evidence $E$, and inference policy $\pi$.
\begin{packeditemize}
  \item \textbf{Inference policy}. $\pi : 2^{\mathit{Evidence}} \times \mathit{Claim} \to [0, 1]$ measurably assigns confidence to evidence--claim pairs.
  Domain-specific $\pi$ is committed in $D$ via $\mathsf{policy\_hash}$.
  \item \textbf{Binary validity}. $\mathsf{SemValid}_\tau(D, E) \equiv \pi(E, \mathsf{claim}(D)) \geq \tau$, with acceptance threshold $\tau \in (0, 1]$.
  \item \textbf{Protocol realization}. $\vinfer$ stochastically estimates $\pi$.
  For honest $p_i$ with $O_i$, $\vinfer$ at $p_i$ samples a distribution centered on $\pi(E, \mathsf{claim}(D))$ with oracle-dependent noise.
  Honest processes may disagree on $\mathsf{SemValid}_\tau(D, E)$ for $\pi(E, \mathsf{claim}(D))$ near $\tau$.
\end{packeditemize}
\end{definition}

\noindent\textbf{Required properties of $\pi$.} \looseness=-1 Each domain-specific \astral instantiation must verify them.
\begin{enumerate}[nosep]
  \item \emph{Measurability}. \looseness=-1 $\pi$ is well-defined on all pairs $(E, \mathsf{claim})$.
  \item \emph{Weak evidence monotonicity}. $\pi(E \cup E', c) \geq \pi(E, c)$ for corroborating evidence $E'$.
  \item \emph{Threshold sensitivity}. Inputs of non-trivial measure satisfy $\pi(E, c) \in (\tau - \beta, \tau + \beta)$ for small $\beta$. Without near-threshold inputs, trivial decidability makes \astral unnecessary.
\end{enumerate}

\astral assumes neither complete $\pi$, sound $\pi$ (false-accept probability $\epsilon$), nor oracle agreement (\Cref{sec:design:probability}).

% --- 4.3.2 Phase 1: Semantic Attestation ---

\subsection{Revised Minimal Explicit Assumptions}\label{app:form:assumptions}

\begin{assumption}[DLS Partial Synchrony]\label{asm:dls}
After unknown GST, all messages between honest replicas arrive within known bound $\Delta$.
View timeouts increase monotonically, so eventually $\mathsf{timeout}(v) > \Delta$ post-GST.
\looseness=-1 Agreement safety is unconditional. Liveness and semantic stabilization hold post-GST.
\end{assumption}

\begin{assumption}[Cryptographic Primitives]\label{asm:crypto}
Assume EUF-CMA signatures and collision-resistant hashes.
All protocol signatures bind $(\mathsf{pid}, \mathsf{inst}, h, v)$ with domain separation against cross-instance replay.
All quorum certificates contain $2f+1$ individual signatures.
\end{assumption}

\begin{assumption}[AFS Availability]\label{asm:afs}
The $\afs$ is trusted for availability.
Two-phase $\daqc$ preserves agreement safety (\Cref{sec:design:daqc}). If $\afs$ is unavailable, $\daqc$ cannot form and liveness stalls.
Each \texttt{DA\_ACK} signature binds the committed Merkle root $\mathsf{afs\_root}$.
Every honest replica independently verifies integrity before acknowledging availability.
Replicating $\afs$ across $2f{+}1$ replicas mitigates selective evidence withholding.
\looseness=-1 Erasure-coded data-availability sampling (DAS) complements replication and could reduce storage without weakening the trust model.
\end{assumption}

\begin{assumption}[Snapshot Determinism]\label{asm:snapshot}
\looseness=-1 All honest replicas evaluate $\vstruct$ on the same committed snapshot, certified by $\daqc$ before any attestation.
\end{assumption}

\begin{assumption}[Base Consensus under Partial Synchrony]\label{asm:a1}
The network satisfies \Cref{asm:dls}.
L1's black-box standard BFT views/epochs and quorum locking prevent conflicting commit/NCC quorums formed by honest replicas for the same slot.
\end{assumption}

\begin{assumption}[Distributional False Accept]\label{asm:a2}
We provide no worst-case bounds for adaptive adversaries crafting inputs to fool $\vinfer$.
For invalid claims from characterized task distribution $\mathcal{D}$, each semantic verifier's empirical false-accept probability is $\Pr[\vinfer = \mathsf{PASS}] \leq \epsilon$.
We measure $\epsilon = \hat{\epsilon}(\tau)$ at each $\tau$.
This risk metric provides no cryptographic soundness guarantee.
\end{assumption}

\begin{assumption}[Diversity with Bounded Correlation (Honest Subset)]\label{asm:a3}
Among $k = 2f+1$ attestations in $\semqc$, at least $k_h = f+1$ come from a fixed set $H$ of $n-f$ honest assessors.
For honest-verifier errors, $\rho_{k_h}$ (\Cref{eq:rho}) is at most 1 under independence, because each honest verifier's rate is at most $\epsilon$, and can exceed 1 under positive correlation.
It has no a priori bound and is measured for each $\mathcal{D}$.
Byzantine attestations are false accepts ($\Pr = 1$) within the $f$ fault budget.
\end{assumption}

\begin{assumption}[Counter-Card Propagation via Mandatory Inclusion]\label{asm:a4}
\astral disseminates counters through $\counterqc$ and Mandatory Inclusion (\Cref{sec:design:challenge}).
Post-GST, partial synchrony with $W(v) \geq 2\Delta$ and the admission cutoff at $W(v) - \Delta$ propagate every admitted counter before NCC signing.
Triage resolves structurally valid, semantically vacuous counters.
$\crqc$ unblocks NCC signing within $O(T_{\mathit{infer}} + \Delta)$. Fallback takes at most $T_{\mathit{arb}}$.
Because every filed counter is delivered post-GST, a challenge fails only if no honest checker detects the invalidity and files a structurally valid counter in time, or if triage dismisses the counter.
\end{assumption}

\begin{assumption}[Verifier Participation]\label{asm:a5}
Post-GST, all honest replicas complete $\vinfer$ within $T_{\mathit{infer}}$, counter triage within $T_{\mathit{arb}}$, and \texttt{DA\_ACK} within $T_{\mathit{DA}}$.
These bounds give $\daqc$ liveness, while $\semqc$ and $\crqc$ formation also depend on honest votes (Remark~\ref{rem:semqcliveness}).
\end{assumption}

\begin{assumption}[Pacemaker Bound]\label{asm:a6}
Post-GST, pacemaker timeouts eventually exceed $\Delta + T_{\mathit{DA}} + T_{\mathit{infer}} + W(v) + T_{\mathit{arb}}$, ensuring stabilization within one stable view.
Alternatively, stabilization state persists across views, preserving all partial progress.
\end{assumption}

\subsection{Stabilization Policy and Proof Sketch}\label{app:form:proof}

\noindent\textbf{Stabilization Policy $\mathcal{P}$.}
\looseness=-1 $D$ becomes $\mathsf{STABLE}$ iff $k$ distinct-replica \texttt{SEM\_PASS} attestations are collected and a $2f+1$ quorum $\ncc$ forms.

\begin{proposition}[Distributional Semantic Failure Decomposition]\label{prop:validity}
Under \Cref{asm:a1,asm:a2,asm:a3,asm:a4,asm:a5,asm:a6}, let $Q$ be the event that at least $k_h$ assessors in $H$ falsely accept. The distributional semantic failure probability satisfies
\begin{equation*}
\delta(\epsilon, k_h) \leq \underbrace{\tbinom{n-f}{k_h}\,\epsilon^{k_h} \rho_{k_h}}_{\geq\, \Pr[Q]} \cdot \underbrace{(1 - p_{\mathit{ch} \mid Q})}_{\text{challenge miss}}.
\end{equation*}
\end{proposition}

\begin{proof}
Failure ($\mathcal{E}$) requires $\semqc$, and hence $Q$ by \Cref{asm:a3}, together with no counter that is filed, delivered, and upheld in triage before $\ncc$ forms.
Thus $\Pr[\mathcal{E}] \leq \Pr[Q]\,(1 - p_{\mathit{ch} \mid Q})$, and $\Pr[Q] = 0$ gives $\Pr[\mathcal{E}] = 0$.
Since $Q$ is the union of $\bigcap_{i \in S} A_i$ over all $S \subseteq H$ with $|S| = k_h$, the union bound and \Cref{eq:rho} give $\Pr[Q] \leq \binom{n-f}{k_h}\epsilon^{k_h}\rho_{k_h}$.
Invalidity that fools $k_h$ verifiers may also evade detection, so unconditional estimates can overstate $p_{\mathit{ch} \mid Q}$.
\end{proof}

% =================================================================
% Appendix B: Agreement Proof Sketches and Censorship Resistance
% =================================================================
\section{Agreement Proof and Censorship Resistance}\label{app:censorship}

% \noindent\textbf{Proof Index.}
% \begin{itemize}[nosep]
%   \item Proposition~\ref{prop:impossibility} (Coverage-Error Tradeoff): \S\ref{app:proof:tradeoff}
%   \item Lemma~\ref{lem:l1safety} (L1 Safety): \S\ref{app:censor:safety}
%   \item Theorem~\ref{thm:stableagreement_app} (Stable Agreement): \S\ref{app:censor:safety}
%   \item Theorem~\ref{thm:astralliveness} (Liveness): \S\ref{app:censor:liveness}
%   \item Theorem~\ref{thm:dissemination_app} (Counter Dissemination): \S\ref{app:censor:mi}
% \end{itemize}

\subsection{False Rejection under Shared Support}\label{app:proof:tradeoff}

\begin{definition}[Semantic Validity Decision (SVD) Task]\label{def:svd}
Each process $p_i$ receives decision object $D = (\mathsf{claim}, \mathsf{evidence})$ and
must output binary judgment $v_i \in \{\mathsf{VALID}, \mathsf{INVALID}\}$.
\begin{itemize}[nosep]
  \item \emph{Agreement.} For all honest processes $p_i, p_j$, $v_i = v_j$.
  \item \emph{Rejection coverage.} The set $F_{\mathit{inv}}$ of non-trivial invalid inputs that the protocol rejects with positive probability is non-empty. Rejection coverage differs from the empirical coverage of \S\ref{sec:eval:setup}, which is the fraction of trials that end in a binary verdict.
  \item \emph{Termination.} All honest processes output with probability 1.
\end{itemize}
\end{definition}

% Full proposition statement is in the main body (\Cref{prop:impossibility}).
% Here we provide the formal SVD task definition, extended discussion, and proof.

\noindent\textbf{Majority example.}
For $n = 7$ assessors that each give the wrong label on a valid input independently with probability $e$, majority false rejection has probability $\sum_{i \geq 4} \binom{7}{i} e^i (1-e)^{7-i}$, which is $2.7 \times 10^{-3}$ for $e = 0.1$ and $3.3 \times 10^{-2}$ for $e = 0.2$.
A disagreement rate does not determine $e$, because two independent calls disagree with probability $2e(1-e)$ for $e$ and for $1-e$ alike.

\smallskip
\noindent\textbf{Proof of Proposition~\ref{prop:impossibility}.}

\begin{proof}
Let $r$ be the coins of $\Pi$, drawn independently of the input and of the oracles.
Given $r$, the queries, stopping, and verdict of $\Pi$ are functions of the outputs observed so far. The set $H_{\mathit{rej}}(r)$ of finite output histories after which $\Pi$ stops with $\mathsf{INVALID}$ therefore does not depend on $D$.
A history $h = (o_1, \ldots, o_t)$ has probability $\Pr_D[h \mid r] = \prod_{j \leq t} \Pr_D[o_j \mid r, o_1, \ldots, o_{j-1}]$ on input $D$, and shared support, which holds for every coin value $r$, makes every factor positive for every reachable $h$ and non-trivial $D$.
Since $\Pi$ terminates with probability 1, rejection is the disjoint union of the events of following some $h \in H_{\mathit{rej}}(r)$, which gives $\Pr[\Pi(D) = \mathsf{INVALID}] = \mathbb{E}_r[\sum_{h \in H_{\mathit{rej}}(r)} \Pr_D[h \mid r]]$ for every non-trivial $D$.
For $D^* \in F_{\mathit{inv}}$ this expectation is positive, hence $H_{\mathit{rej}}(r)$ is non-empty for coins of positive probability.
For a non-trivial valid $D'$ the inner sum is positive on those coins, which yields $P_{\mathit{err}}(\Pi, D') > 0$.

\emph{Examples.}
One round of majority voting over independent calls with valid-input error $e > 0$ falsely rejects with the positive binomial tail above.
Unanimous $\mathsf{INVALID}$ lowers false rejection to $e^n$, but a call that rejects an invalid input with probability $q$ then yields rejection only with probability $q^n$.
Both rules have the same rejection coverage and differ only in their rejection probabilities.
A protocol that always outputs $\mathsf{VALID}$ has $P_{\mathit{err}} = 0$ and empty $F_{\mathit{inv}}$.
\looseness=-1 A deterministic $\vstruct$ reads $D$ directly and lies outside the hypothesis. It agrees without false rejection, but it checks only that evidence exists and is well-formed and cannot test whether the evidence \emph{entails} the claim.
\end{proof}

\subsection{Stabilization-Layer Liveness (Post-GST, Conditional)}\label{app:censor:liveness}

\noindent\textbf{SemQC Liveness.}
Under static corruptions with $n = 3f+1$ and $f$ Byzantine replicas, $2f+1$ honest replicas run $\vinfer$.
Let $\epsilon_{\mathit{fn}}$ denote honest-verifier \textsf{REJECT} probability on valid claims (false-negative rate).
$\semqc$ requires $k = 2f+1$ distinct-replica attestations, and $f_{\mathit{attest}} \leq f$ counts Byzantine attestations that arrive in time.
If all $f$ Byzantine replicas withhold, every honest replica must pass, which is guaranteed only when $\epsilon_{\mathit{fn}} = 0$.
If all $f$ Byzantine replicas attest in time, $f+1$ honest passes suffice, and the expected honest pass count $(1 - \epsilon_{\mathit{fn}})(2f+1)$ reaches $f+1$ when
\begin{equation}\label{eq:fn_constraint}
  \epsilon_{\mathit{fn}} \leq \frac{f}{2f+1} < \frac{1}{2}
\end{equation}
This bounds only the expected count. For independent honest errors and $f_{\mathit{attest}}$ fixed independently of them, $\semqc$ forms with probability $\Pr[X \geq 2f+1-f_{\mathit{attest}}]$ for $X \sim \mathrm{Bin}(2f+1, 1-\epsilon_{\mathit{fn}})$.
By \Cref{asm:a5}, honest attestations arrive within $O(T_{\mathit{infer}} + \Delta)$.
\looseness=-1 \textbf{Fallback.} Lowering $k$ to $f+1$ weakens certification to $k_h = 1$ and higher $\delta$.

\begin{remark}[SemQC Liveness is Conditional]\label{rem:semqcliveness}
Under withholding ($f_{\mathit{attest}} = 0$), $\semqc$ at $k = 2f+1$ forms with certainty only when $\epsilon_{\mathit{fn}} = 0$.
With all Byzantine replicas attesting in time ($f_{\mathit{attest}} = f$), $\epsilon_{\mathit{fn}} \leq f/(2f+1)$ (\Cref{eq:fn_constraint}), or $2/5$ at $n = 7$, brings the expected count to $k$, and formation remains probabilistic.
Choose $k$ by threat model, between $k = 2f+1$ (strong, conditional on $\epsilon_{\mathit{fn}}$) and $k = f+1$ (minimal, tolerates up to $f$ honest rejections, but $k_h = 1$).
\end{remark}

\smallskip
\noindent\textbf{NCC Liveness.}
\looseness=-1 Without challenges, after $W(v) \geq 2\Delta$ all honest replicas sign \texttt{NO\_CHALLENGE} $\to$ $\ncc$ within $O(W(v) + \Delta)$.
For challenged claims, triage resolves every counter or flags \texttt{DISPUTED} within worst-case $O(f \cdot T_{\mathit{arb}})$.

\smallskip
\noindent\textbf{StableQC Liveness.}
Post-GST, valid unchallenged claims obtain $\stableqc$ within $O(2\Delta + W(v) + T_{\mathit{infer}} + T_{\mathit{DA}})$ after proposal once $\semqc$ forms within its attestation bound.
Challenges add triage time $O(f \cdot T_{\mathit{arb}})$.
Unresolved counters trigger \texttt{DISPUTED}.

\begin{theorem}[\astral Liveness (Post-GST)]\label{thm:astralliveness}
Under \Cref{asm:a1,asm:a2,asm:a3,asm:a4,asm:a5,asm:a6}, the following hold post-GST for any height $h$.
\begin{enumerate}[nosep]
  \item \textbf{Commit liveness.} $\mathsf{Commit}(h)$ is emitted within $O(f)$ view-changes $+$ $O(\Delta + T_{\mathit{DA}})$ per view (inherited from L1).
  \item \textbf{Stabilization liveness.} If at least $2f+1$ \texttt{SEM\_PASS} attestations arrive within $O(T_{\mathit{infer}} + \Delta)$ (Remark~\ref{rem:semqcliveness}) and no counter carries a valid deterministic fraud-proof, $\stableqc$ forms within $O(2\Delta + W(v) + T_{\mathit{infer}} + T_{\mathit{DA}})$ after $\commitqc$.
  For challenged claims, counter triage adds $O(f \cdot T_{\mathit{arb}})$.
  \looseness=-1 Counters unresolved after triage flag the decision \texttt{DISPUTED} within $O(W(v) + T_{\mathit{arb}})$.
\end{enumerate}
\end{theorem}
\begin{proof}
L1 gives commit liveness (\textbf{Appendix~\ref{app:hotstuff:liveness}}), since \astral only restricts voting (Remark~\ref{rem:extvalidity:app}) and an honest leader can satisfy all conditions post-GST with AFS availability (\Cref{asm:afs}). Deterministic counter forwarding converges within $\Delta$, satisfying Mandatory Inclusion.
$\semqc$ formation under this attestation condition and NCC liveness with resolved counters or \texttt{DISPUTED} timeout yield stabilization within one stable view (\Cref{asm:a6}).
\end{proof}

\section{L1 Interface}\label{app:hotstuff}

\astral requires totally ordered $\mathsf{Commit}(h)$, at most one per height, and post-GST commit liveness (\S\ref{sec:design:l1}).
Before voting, L1 must evaluate external validity $P_{\mathit{ext}}(b) \to \{\mathsf{accept}, \mathsf{reject}\}$.
The L2 admission checks $\vstruct$, $\daqc$ verification, and Mandatory Inclusion run after $\commitqc$ (\Cref{alg:consensus}), so \astral adds no evidence-dependent condition to $P_{\mathit{ext}}$ and L1 votes never wait on evidence.
With deterministic BFT (HotStuff~\cite{hotstuff}, PBFT~\cite{pbft}, or Tendermint~\cite{tendermint}), L2 starts at $\commitqc$.
\looseness=-1 On probabilistic chains, L2 starts at $k$-deep confirmation. Semantic failure or reversion has probability at most $\delta + p_{\mathit{rev}}(k)$ (\S\ref{sec:eval:rq3}).

\begin{remark}[External Validity Preserves Safety]\label{rem:extvalidity:app}
\looseness=-1 Quorum intersection and honest non-equivocation preserve L1 safety when $P_{\mathit{ext}}$ restricts accepted proposals.
\end{remark}

\smallskip
\noindent\textbf{Liveness.}\label{app:hotstuff:liveness}
Post-GST, an honest leader can satisfy all \astral conditions given available evidence (\Cref{asm:afs}) and deterministic $\counterqc$ forwarding for Mandatory Inclusion.
Post-GST, once an unresolved $\counterqc$ reaches all honest assessors (\Cref{thm:dissemination}), a block omitting it obtains no $\daqc$ until the counter is included or resolved.
\looseness=-1 Stabilization requires bounded-time $\vinfer$ (\Cref{asm:a5}) and every valid counter's resolution within $T_{\mathit{arb}}$ (\S\ref{sec:design:challenge}).

% =================================================================
% Appendix D: Complete Protocol Specification
% =================================================================
\section{Complete Protocol Specification}\label{app:protocol}

\smallskip
\noindent\textbf{Window cap.}
Operators cap $W(v)$ at $W_{\max}$ to bound worst-case stabilization latency, \eg 60\,s interactively and 300\,s for batch pipelines.

\subsection{Formal Pipeline (Pseudocode)}\label{app:protocol:pseudocode}

\Cref{alg:consensus} specifies L1/L2 ($\daqc$), and \Cref{alg:stabilization} specifies L3 attestation, triage, and certification.

\begin{algorithm}[!htbp]
\caption{\astral Consensus Core (Replica $i$)}\label{alg:consensus}
\scriptsize
\begin{algorithmic}[1]
\STATE \textbf{Params:} $n, f, \mathsf{quorum}{=}2f{+}1, k{=}\mathsf{quorum}, \tau$; $W(v) = \max(2\Delta{+}\epsilon_{\mathit{proc}},\, W_0 {\cdot} v)$
\STATE \textbf{Leader Propose:} Build $b{=}(P, \mathsf{counter\_batch}, \mathsf{payload})$, broadcast $\texttt{PROPOSE}(\ds, b, \mathsf{highQC})$
\STATE \textbf{Replica Vote:} if valid and unvoted: log; $\mathsf{voted}[h,v] \gets \mathsf{true}$; send $\texttt{VOTE}(\ds, H(b))$
\STATE \textbf{Commit:} Leader forms QC on $\geq\mathsf{quorum}$ votes; on 3-chain: $\commitqc$; output $\mathsf{Commit}(h)$
\STATE \textbf{L2 Admission:} On $\commitqc$: if $\vstruct(D(b)) = \mathsf{PASS}$: broadcast $P$; collect $2f{+}1$ $\texttt{DA\_ACK} \to \daqc$; open window; trigger \Cref{alg:stabilization}. Mandatory Inclusion: deny \texttt{DA\_ACK} if known $\counterqc$s omitted.
\STATE \textbf{View-Change:} If no progress within $\mathsf{timeout}(v)$: $v \gets v{+}1$
\end{algorithmic}
\end{algorithm}

\begin{algorithm}[!htbp]
\caption{\looseness=-1 \astral Stabilization \& Counter Triage (Replica $i$)}\label{alg:stabilization}\label{alg:triage}
\scriptsize
\begin{algorithmic}[1]
\STATE \textbf{Attestation:} $(\mathsf{label}, \mathsf{conf}, \mathsf{ref}) \gets \vinfer(D(b))$; if $\mathsf{SUPPORT}$ and $\mathsf{conf} \geq \tau$: send \texttt{SEM\_PASS}
\STATE \textbf{SemQC:} On $\geq k$ distinct \texttt{SEM\_PASS}: form and forward $\semqc$
\STATE \textbf{Counter Triage:} For each $\counterqc$ of a counter filed before $W(v) - \Delta$:
\STATE \quad Deterministic fraud-proof: re-execute in WASM ($\leq G_{\max}$); valid $\to$ SUSTAINED ($\stableqc$ blocked); invalid $\to$ \texttt{COUNTER\_REJECT}
\STATE \quad Semantic counter: $\vinfer$ triage; refutes $\to$ \texttt{COUNTER\_SUSTAIN}; else $\to$ \texttt{COUNTER\_REJECT}
\STATE \textbf{CRQC:} $\geq\mathsf{quorum}$ \texttt{COUNTER\_REJECT} $\to$ $\crqc$; remove from $\mathsf{unresolvedCounters}$
\STATE \textbf{Closure:} After $W(v)$: if $\mathsf{unresolvedCounters} = \emptyset$: send \texttt{NO\_CHALLENGE}; $\geq\mathsf{quorum}$ $\to$ $\ncc$
\STATE \textbf{Timeout:} If counters remain unresolved at $W(v) + T_{\mathit{arb}}$: provisional $\mathsf{disputed}[h] \gets \mathsf{true}$
\STATE \textbf{Certify:} If $\commitqc, \semqc, \ncc \neq \bot$, no SUSTAINED: $\mathsf{Certify}(h)$, clear $\mathsf{disputed}[h]$
\STATE \textbf{Reject:} After $W(v){+}\Delta$, if $\semqc = \bot$, $\mathsf{unresolvedCounters} = \emptyset$, $\neg\mathsf{disputed}[h]$, no SUSTAINED: provisional $\mathsf{Reject}(h)$ until a forwarded $\semqc$ arrives
\end{algorithmic}
\end{algorithm}

\begin{invariant}[Stabilization--L1 Isolation]\label{inv:isolation}
Stabilization state (L3) is \textbf{monotone} and \textbf{block-scoped}.
L3 state never enters L1 voting. It reaches L2 only through Mandatory Inclusion, which reads unresolved certified counters and can only withhold \texttt{DA\_ACK}.
Voting consults only $P_{\mathit{ext}}$, $\mathsf{voted}[h,v]$, and $\mathsf{lockedQC}$.
Thus no stabilization certificate can induce a vote violating $\mathsf{lockedQC}$.
\end{invariant}

\begin{invariant}[Input Consistency for L3]\label{inv:inputconsistency}
Before any L3 attestation, all honest replicas commit the same L1 block ($\commitqc$ ensures identical $\mathit{val}_h$), verify structural validity with identical deterministic $\vstruct$ results, and obtain $\daqc$ certifying the same evidence snapshot's availability.
All honest replicas thus have identical, verified input for probabilistic semantic reasoning, while deterministic L1/L2 safety is independent of non-deterministic L3.
\end{invariant}

\noindent\textbf{Attestation threshold and diversity.}\label{app:protocol:threshold}\label{app:protocol:diversity}
Evaluation uses $k = 2f{+}1$. Lowering $k$ (down to $f{+}1$) tolerates more false negatives but weakens $\delta$ bounds.
\looseness=-1 Diversity predicates are recommended (\S\ref{sec:design:probability}), but optional for $\semqc$.

% =================================================================
% Appendix E: Evaluation on Real Public Claim-Verification Benchmarks
% =================================================================
\section{Supplementary Evaluation Results}\label{app:Supplementary_Evaluation}

\subsection{Committee, Non-Web Domains, and Additional Metrics}\label{app:extra}
\looseness=-1 GPT-4o, Claude Haiku 4.5, and Gemini 3 Flash are accessed via API. Llama 3.1 70B, Mixtral 8$\times$7B, Qwen 2.5 32B, and DeepSeek-R1-Distill 32B run locally, 4-bit quantized.
The $n{=}7$ committee uses two replicas per API family and one Llama 3.1 70B. At $n{=}13$, each remaining local family adds two replicas.
Non-Web evaluation uses 60 synthetic Python snippets across six CWE families (code vulnerability patching, CVP) and 50 synthetic syslog bundles across five incident types (semantic log auditing, SLA), half fabricated in each domain.
Under A4, \astral accepts no refuted SLA (25 items) or CVP (30 items) claim at any $n \in \{7, 13, 31, 67, 100\}$.
MV accepts none at $n{=}7$ but 88.0\% (SLA) and 46.7\% (CVP) at every larger $n$. CP-WBFT accepts 12.0--32.0\% (SLA) and 0.0--13.3\% (CVP).
Under A1--A4 at $n{=}7$, counters are filed in 56.7--59.0\% of \astral trials, a $\crqc$ resolves them in 5.0--8.9\%, and 53.8--55.7\% end \texttt{DISPUTED}.
\astral's protocol correctness excludes \texttt{DISPUTED} verdicts and is 42.4--44.8\%.

% Sources: artifacts/analysis/www_public_w1.txt section 2 (public rows),
% artifacts/analysis/selective_mv_matched_coverage.txt lines 13, 15, 19 (Synthetic row).
\begin{table}[!htbp]
\centering
\caption{Abstaining majority voting on \astral's votes (\%). Conf.\ MV abstains below a mean-confidence threshold calibrated to \astral's benign coverage, and MV$_5$ needs five agreeing votes. TAR is the acceptance rate of supported claims under A1--A4. Synthetic pools MFV, CVP, and SLA at $n{=}7$.}\label{tab:matched_cov}
\setlength{\tabcolsep}{4pt}
\renewcommand{\arraystretch}{0.88}
\resizebox{\columnwidth}{!}{%
\begin{tabular}{@{}l*{9}{r}@{}}
\toprule
 & \multicolumn{3}{c}{\textbf{\astral}} & \multicolumn{3}{c}{\textbf{Conf.\ MV}} & \multicolumn{3}{c}{\textbf{MV$_5$}} \\
\cmidrule(lr){2-4}\cmidrule(lr){5-7}\cmidrule(l){8-10}
 & Cov. & ASR$\downarrow$ & TAR$\uparrow$ & Cov. & ASR$\downarrow$ & TAR$\uparrow$ & Cov. & ASR$\downarrow$ & TAR$\uparrow$ \\
\midrule
AVeriTeC & 16.0 & \textbf{21.5} & 57.8 & 16.0 & 59.1 & \textbf{82.4} & 82.4 & 44.3 & 70.9 \\
Climate-FEVER & 37.5 & \textbf{43.0} & 65.4 & 37.1 & 71.4 & \textbf{85.7} & 73.6 & 53.5 & 73.4 \\
FEVER & 33.3 & \textbf{37.0} & 77.4 & 33.0 & 69.3 & \textbf{88.6} & 68.6 & 54.9 & 85.1 \\
PUBHEALTH & 55.9 & \textbf{36.1} & 60.6 & 55.9 & 54.7 & \textbf{80.3} & 80.7 & 43.3 & 65.8 \\
SciFact & 61.7 & \textbf{53.1} & 67.9 & 61.7 & 77.7 & \textbf{84.9} & 58.3 & 61.1 & 72.2 \\
Synthetic & 43.8 & 1.4 & 86.4 & 41.0 & \textbf{0.0} & 54.3 & 99.0 & 2.6 & \textbf{89.0} \\
\bottomrule
\end{tabular}}
\end{table}

\looseness=-1 Forged votes carry confidence 0.95, so Conf.\ MV accepts 54.7--77.7\% of refuted public-benchmark claims.
On the synthetic suite, it accepts no refuted claim and 54.3\% of supported ones, against 86.4\% for \astral.
MV that needs six or seven agreeing votes has lower ASR than MV$_5$, but its TAR stays below \astral's on every public benchmark, and two Byzantine assessors can block any of its decisions.
With refuted claims resampled, the 95\% interval of \astral's pooled-ASR gap to each bracketed baseline in \Cref{tab:astral_main} excludes zero.

\looseness=-1 Under A2, ASR on CVP and MFV is 100\% for L1-BFT, 0.0\% and 6.0\% with a $\semqc$ from $\vstruct$ and $\vinfer$, and 0.0\% and 2.0\% for full \astral.
\astral-NoDiv averages 3.3\% and 3.5\% over backends (CVP 0.0--13.3\%, MFV 2.0--4.0\%).
Simple Threshold, which follows one seeded assessor per trial, reaches 23.3\% and 40.0\%.

\begin{table}[!htbp]
\centering
\caption{Synthetic suite under A1--A4 (840 attacked, 210 benign trials per protocol and $n$). Attack Reject.\ counts refuted claims not accepted, and tokens are on code patching. L1 = L1-BFT, CP = CP-WBFT.}\label{tab:protocol-comparison}
\evaltablefont
\setlength{\tabcolsep}{7pt}
\renewcommand{\arraystretch}{0.88}
\resizebox{\columnwidth}{!}{
\begin{tabular}{@{}lcccc|cccc@{}}
\toprule
 & \multicolumn{4}{c|}{$n = 7$} & \multicolumn{4}{c}{$n = 31$} \\
\textbf{Metric (\%)} & \textbf{L1} & \textbf{MV} & \textbf{CP} & \textbf{\astral} & \textbf{L1} & \textbf{MV} & \textbf{CP} & \textbf{\astral} \\
\midrule
Attack Reject.\ $\uparrow$ & 0.0 & 96.9 & 82.1 & \astralcell{98.6} & 0.0 & 81.7 & 79.5 & \astralcell{99.0} \\
Benign SCSR $\uparrow$ & 50.0 & 93.8 & \astralsoft{\textbf{94.3}} & 43.3 & 50.0 & 93.8 & \astralsoft{\textbf{94.3}} & 32.9 \\
Counter-Card Filing & -- & -- & -- & \astralsoft{58.3} & -- & -- & -- & \astralsoft{67.0} \\
$\crqc$ Resolution $\uparrow$ & -- & -- & -- & \astralsoft{6.3} & -- & -- & -- & \astralsoft{13.5} \\
Stable Finality $\uparrow$ & -- & -- & -- & \astralsoft{43.9} & -- & -- & -- & \astralsoft{41.0} \\
Tokens/decision & -- & 8.1K & 16.4K & \astralsoft{12.0K} & -- & 36.9K & 180.0K & \astralsoft{172.7K} \\
\bottomrule
\end{tabular}
}
\end{table}

\clearpage

\end{document}